\PassOptionsToPackage{dvipsnames}{xcolor}
\documentclass[acmsmall,screen,nonacm]{acmart}

\input{preamble.tex}

\providecommand{\textcite}{\citet}
\providecommand{\autocite}{\cite}

\begin{document}

\title{Bidirectional Type Slicing}

\author{Max Carroll}
\orcid{0009-0007-9636-8502}
\affiliation{
  \institution{University of Cambridge}
  \department{Department of Computer Science and Technology}
  \city{Cambridge}
  \country{UK}}
\author{Anil Madhavapeddy}
\orcid{0000-0001-8954-2428}
\affiliation{
  \institution{University of Cambridge}
  \department{Department of Computer Science and Technology}
  \city{Cambridge}
  \country{UK}}
\author{Cyrus Omar}
\orcid{}
\affiliation{
  \institution{University of Michigan}
  \city{Ann Arbor}
  \country{USA}}

\begin{abstract}
Development tools report \emph{what} type an expression has, but not \emph{why} it has that type. This paper develops a theory of \emph{type slicing} that answers such questions: a programmer selects a term, queries any part of the type information associated with it, and receives a slice of the program---a well-formed partial program with irrelevant sub-terms folded away---that suffices to reproduce the queried type. We formulate type slicing for bidirectional type systems, where \emph{synthesis slices} explain the type a term synthesises and \emph{analysis slices} explain the type its surrounding context expects. The theory requires no cast dynamics as it applies to any bidirectional system equipped with a precision order on types and terms satisfying a downwards static graduality property. We develop the metatheory over a core calculus with holes, products, sums, and explicit polymorphism, based on the Hazelnut and marked lambda calculi, proving that every query has a minimal slice, and that refining a query monotonically shrinks its minimal slices. We then show how to calculate these slices both exactly and approximately. Finally, integrating type slicing with error marking theory extends these results to arbitrary ill-typed programs, so a single mechanism explains both types and type errors in \textit{complete}, \textit{incomplete}, and \textit{erroneous} code. The metatheory is mechanised in Agda, and a linear-time approximation of type slicing is implemented for the Hazel programming environment.
\end{abstract}

\maketitle

\section{Introduction}
\label{sec:intro}
\label{firstcontentpage}

Programming often involves reasoning about static types.
Development tools assist programmers in limited ways with this task, 
e.g. by reporting the type of an expression when hovering over it, or reporting the expected and actual type when there is a type error.
However, these services only report \emph{what} the type of an expression is, not \emph{why} it is that type 
in the context of the surrounding program.
Programmers must hunt through the program themselves to answer these questions, and the program often contains a large volume of code irrelevant to their query.

This paper develops the theory of \textit{type slicing} from type-theoretic first principles. Type slices allow developers to ask ``why questions''~\cite{developers-ask-why-questions} about type information reported about arbitrary terms within a (possibly incomplete or erroneous) program. The answers are given in the form of a minimal slice of the original program (i.e. a program with unrelated terms folded away) that suffices to reproduce the queried type. 
Type slicing is inspired in principle by prior work on program slicing (\S\ref{sec:related-work}), in particular, lattice-based \textit{Galois slicing}~\cite{Perera2012,Ricciotti2017,PereraGarg2016,PereraVis2022,AtkeyPerera2025}. However, slicing is here applied to a (non-Galois) static type checking setting, loosely more akin to type-error slicing~\cite{Wand1986,HaackWells2003,TipDinesh2001}, and the concurrent independently researched work on \emph{type-directed slicing}~\cite{TypeDirectedSlicing} for Java.

We consider the problem of type slicing in the context of \emph{bidirectional type systems}~\cite{BidirectionalTypes}, which distinguish \emph{type synthesis} (which type does this term have) from \emph{type analysis} (which type is expected).
Correspondingly, a \textit{synthesis slice} (\S\ref{sec:synthesis}) of a selected term answers ``which parts of the program caused it to have type $\tau$?'' and an \textit{analysis slice} (\S\ref{sec:analysis}) answers the complementary question, ``which parts of the surrounding program enforced that this term have type $\tau$?''
Formalising analysis slices requires a novel \textit{context typing} construction (\S\ref{sec:context-typing}), which types the one-hole context surrounding a focused sub-term.
Queries can be \textit{refined}, say from an entire function type to just its codomain, resulting in a smaller slice. This allows a programmer to incrementally focus on sub-parts of the type that are of interest, for example, determining that only a small region of a large record type involved in a type error is erroneous.

\paragraph{Type Error Debugging} The type slicing machinery applies equally well to well-typed and ill-typed terms. In the case of ill-typed terms, it provides a mechanism that may aid in debugging type errors, subsuming the question of ``why is this term a type error?'' This is achieved by integrating type slicing with the theory of error marking (\S\ref{sec:marking}).
A distinguishing feature of type slicing is that a slice is itself a valid incomplete program that preserves the original's shape; in Hazel~\cite{hazel-website}, which allows incomplete programs, this is simply a valid program. Type slicing therefore \emph{composes} with most other type debugging methods, as we can simply run them on the minimised programs. 

For example, error \emph{localisation}~\cite{ZhangMyers2014,PavlinovicKingWies2014} and \emph{counter-factual}~\cite{ChenErwig2014,LernerSeminal2007,Chitil2001} type debugging techniques (\S\ref{sec:cmp-debuggers}) are equally applicable, and in fact simpler, on minimised programs, when their language assumptions (typically global inference) apply.

Further, editing the `type' of an expression as suggested by such tools is not necessarily easy when the surrounding context expects a different type. The \textit{analysis slice} points to the surrounding terms (typically annotations) that must be modified to make such an edit, likewise easing refactoring to a different type in general.

\vspace{-6pt}
\subsection{Contributions}
\label{sec:scope}
The theory applies to \textit{bidirectional type systems} with any form of \textit{downwards static graduality} on expressions (\S\ref{sec:background}), inspired by the gradual typing literature~\cite{GradualRefined} extended to incomplete programs, i.e. programs with holes (and marked errors). No cast dynamics are required, so type slicing applies to any bidirectional system upon assigning a precision relation on types and terms (\S\ref{sec:core-lattice}).

Concretely, this paper contributes:
{\em (i)} a core calculus with holes, products, sums, and explicit polymorphism, based on the Hazelnut and marked lambda calculi~\cite{Hazelnut2017,Marking2024}, with precision lattices satisfying graduality (\S\ref{sec:core-calculus});
{\em (ii)} \textit{synthesis slices} (\S\ref{sec:synthesis}), whose minimal slices exist for every query, refine monotonically, and decompose compositionally;
{\em (iii)} \textit{context typing} (\S\ref{sec:context-typing}), a novel judgement with totality, soundness, and gradual guarantees, enabling the dual \textit{analysis slices} (\S\ref{sec:analysis});
{\em (iv)} integration with error marking (\S\ref{sec:marking}), extending both slice forms to arbitrary ill-typed programs;
{\em (v)} term-minimal slices supporting compositional calculation, NP-hardness of minimum-size slicing, and a linear-time approximation (\S\ref{sec:algorithms}--\S\ref{sec:min-size});
and {\em (vi)} a comparison with prior forms of slicing and type debugging, and applications to set-theoretic and occurrence typing (\S\ref{sec:related-work}).

Hazel~\cite{hazel-website} is an example full-scale language satisfying these conditions; we have extended its implementation with an efficient linear-time approximation of type slicing, introduced by example (\S\ref{sec:slicing-by-example}). The prototype user interface is an initial proof of concept, and we make no specific usability claims; we expect the theory to support a variety of user interfaces and further applications for type debugging and explanation, and suggest several avenues for exploration (\S\ref{sec:further-work}).
The core calculus and metatheory are mechanised in Agda and available in the supplemental material.

\subsection{Type Slicing in Hazel by Example}
\label{sec:slicing-by-example}

{This section previews type slicing, formally developed later (\S\ref{sec:synthesis}, \S\ref{sec:analysis}) over the core calculus (\S\ref{sec:core-calculus}), through two worked examples in the Hazel implementation.
The Hazel UI has a cursor inspector at the bottom of the screen which details which type the selected term \textit{has}, and what it \textit{expects}; these are the respective types to calculate a \synsrc{synthesis} slice and an \anasrc{analysis} slice for.

Hazel renders each slice by folding away every fragment the queried type does not require. The user can refine their slice by themselves folding away part of the type being sliced, which will correspondingly fold away additional fragments of the program.}

\definecolor{annotcolor}{HTML}{955BE3}
\tikzset{
  cobase/.style={font=\scriptsize\sffamily, align=left, inner sep=2pt,
    rounded corners=1.5pt, fill=white, fill opacity=0.9, text opacity=1,
    line width=0.5pt},
  synco/.style={cobase, draw=NavyBlue, text=NavyBlue!75!black},
  anaco/.style={cobase, draw=OliveGreen!70!black, text=OliveGreen!45!black},
  foldco/.style={cobase, draw=black!45, text=black!55},
  errco/.style={cobase, draw=BrickRed, text=BrickRed!85!black},
  purpco/.style={cobase, draw=annotcolor, text=annotcolor!80!black},
  synl/.style={-{Stealth[length=1.5mm]}, densely dotted, line width=0.7pt, NavyBlue},
  anal/.style={-{Stealth[length=1.5mm]}, densely dotted, line width=0.7pt, OliveGreen!70!black},
  foldl/.style={-{Stealth[length=1.5mm]}, densely dotted, line width=0.7pt, black!45},
  errl/.style={-{Stealth[length=1.5mm]}, densely dotted, line width=0.7pt, BrickRed},
  purpl/.style={-{Stealth[length=1.5mm]}, densely dotted, line width=0.7pt, annotcolor},
}
\newcommand{\foldarrow}[1]{\par\vspace{0.7mm}
  {\footnotesize\sffamily\color{black!55}$\big\Downarrow$\ \ #1}\par\vspace{0.7mm}}

\begin{figure}[t]
\vspace{6pt}
\centering

\begin{tikzpicture}
\node[anchor=south west, inner sep=0] (img) at (0,0)
  {\includegraphics[width=0.9\linewidth]{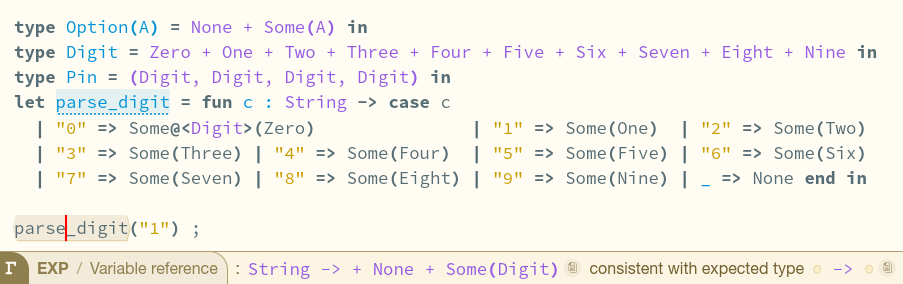}};
\draw[black!35, line width=0.35pt] (img.south west) rectangle (img.north east);
\begin{scope}[x={(img.south east)},y={(img.north west)}]
  \node[synco] (a1) at (0.24,0.29) {select \texttt{parse\_digit}};
  \draw[synl] (a1.west) -- (0.08,0.25);
  \node[synco] (a2) at (0.52,0.2) {which synthesises this type};
  \draw[synl] (a2.south) -- (0.43,0.07);
\end{scope}
\end{tikzpicture}

\foldarrow{fold to the synthesis slice}

\begin{tikzpicture}
\node[anchor=south west, inner sep=0] (img) at (0,0)
  {\includegraphics[width=0.9\linewidth]{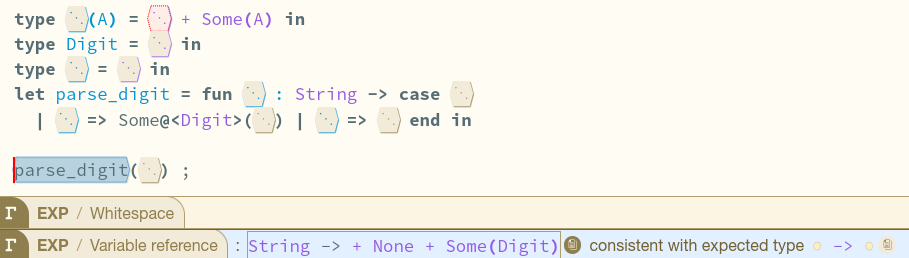}};
\draw[black!35, line width=0.35pt] (img.south west) rectangle (img.north east);
\begin{scope}[x={(img.south east)},y={(img.north west)}]
  \node[synco] (b1) at (0.62,0.23) {press slicing icon};
  \draw[synl] (b1.south) -- (0.63,0.08);
\end{scope}
\end{tikzpicture}

\caption{Synthesis slice for \synsrc{\texttt{parse\_digit}}.}
\label{fig:parse-digit-slices}
\label{fig:parse-digit-synthesis}
\end{figure}

\Cref{fig:parse-digit-synthesis} slices a small character parser. The synthesis slice reduces the function to the \texttt{String} domain annotation and the first branch determining the result option type, discarding all other branches.\footnote{Folding all other branches into a single fold projector is not currently supported by the Hazel UI, but is shown here for illustrative purposes.}

Refined queries may target \emph{part} of a reported type: for example, asking only about the domain of a function type, by folding away the returned sum type (\synsrc{\texttt{String -> }$\fold$}). The resulting program shrinks the slice to exclude the entire function body (\cref{fig:parse-digit-analysis}). This same slice is also the minimal \textit{analysis} slice which \textit{enforces} \anasrc{\texttt{String}} on the argument.
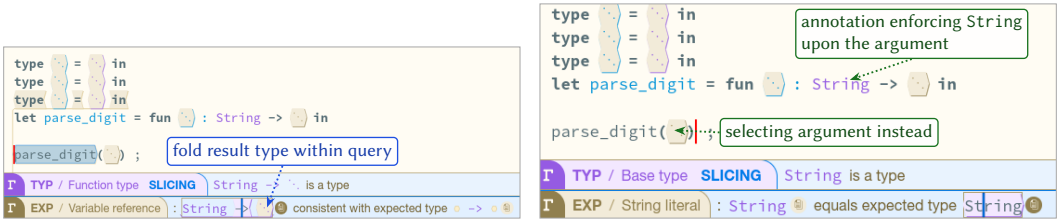
\begin{figure}[h]
\centering

\begin{subfigure}[t]{0.49\linewidth}
\centering
\begin{tikzpicture}
\node[anchor=south west, inner sep=0] (img) at (0,0)
  {\includegraphics[width=\linewidth]{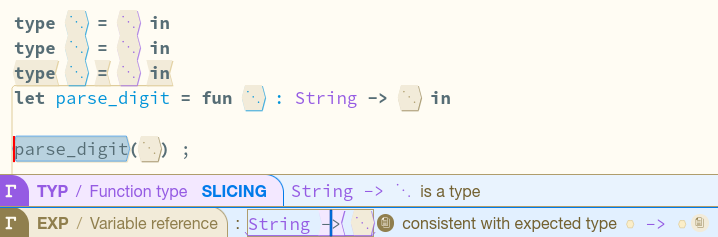}};
\draw[black!35, line width=0.35pt] (img.south west) rectangle (img.north east);
\begin{scope}[x={(img.south east)},y={(img.north west)}]
  \node[synco] (d1) at (0.54,0.39) {fold result type within query};
  \draw[synl] (d1.south) -- (0.51,0.062);
\end{scope}
\end{tikzpicture}
\caption{Refined synthesis query omitting the return Option type.}
\label{fig:parse-digit-domain-query}
\end{subfigure}\hfill
\begin{subfigure}[t]{0.49\linewidth}
\centering
\begin{tikzpicture}
\node[anchor=south west, inner sep=0] (img) at (0,0)
  {\includegraphics[width=\linewidth]{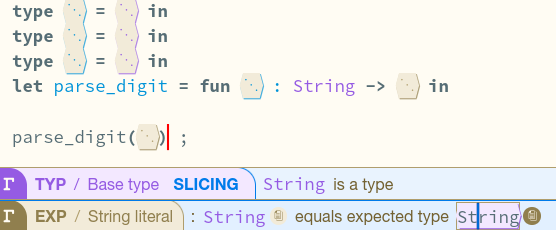}};
\draw[black!35, line width=0.35pt] (img.south west) rectangle (img.north east);
\begin{scope}[x={(img.south east)},y={(img.north west)}]
  \node[anaco] (c1) at (0.56,0.4) {selecting argument instead};
  \draw[anal] (c1.west) -- (0.26,0.41);
  \node[anaco] (c2) at (0.72,0.86) {annotation enforcing \texttt{String}\\upon the argument};
  \draw[anal] (c2.south) -- (0.60,0.645);
\end{scope}
\end{tikzpicture}
\caption{Analysis slice of the \anasrc{\texttt{String}} expectation on the argument of \synsrc{\texttt{parse\_digit}}.}
\label{fig:parse-digit-argument-analysis}
\end{subfigure}

\caption{Refined synthesis and analysis slices.}
\label{fig:parse-digit-analysis}
\end{figure}

\clearpage

\begin{figure}[t]
\centering

\begin{tikzpicture}
\node[anchor=south west, inner sep=0] (img) at (0,0)
  {\includegraphics[width=0.9\linewidth, trim={0 0 0 62}, clip]{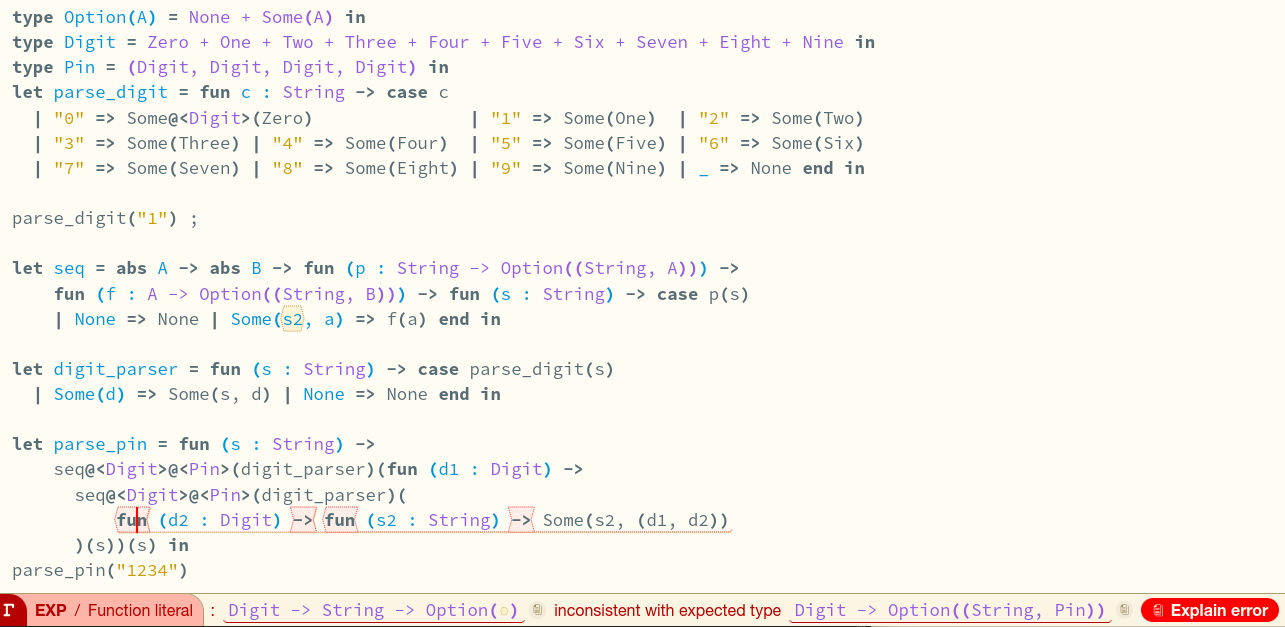}};
\draw[black!35, line width=0.35pt] (img.south west) rectangle (img.north east);
\begin{scope}[x={(img.south east)},y={(img.north west)}]
  \node[errco] (e1) at (0.80,0.63) {selected curried function synthesises\\
    \texttt{Digit -> String -> Option(\ldots)}};
  \draw[errl] (e1.south west) -- (0.4,0.202);
  \node[anaco] (e2) at (0.8,0.4) {\ldots but its position demands\\
    \texttt{Digit -> Option((String, Pin))}};
  \draw[anal] (e2.south) -- (0.68,0.040);
  \node[errco] (e3) at (0.85,0.150) {press \emph{Explain error}};
  \draw[errl] (e3.south) -- (0.92,0.046);
\end{scope}
\end{tikzpicture}

\foldarrow{explain the error}

\begin{tikzpicture}
\node[anchor=south west, inner sep=0] (img) at (0,0)
  {\includegraphics[width=0.9\linewidth, trim={0 0 0 57}, clip]{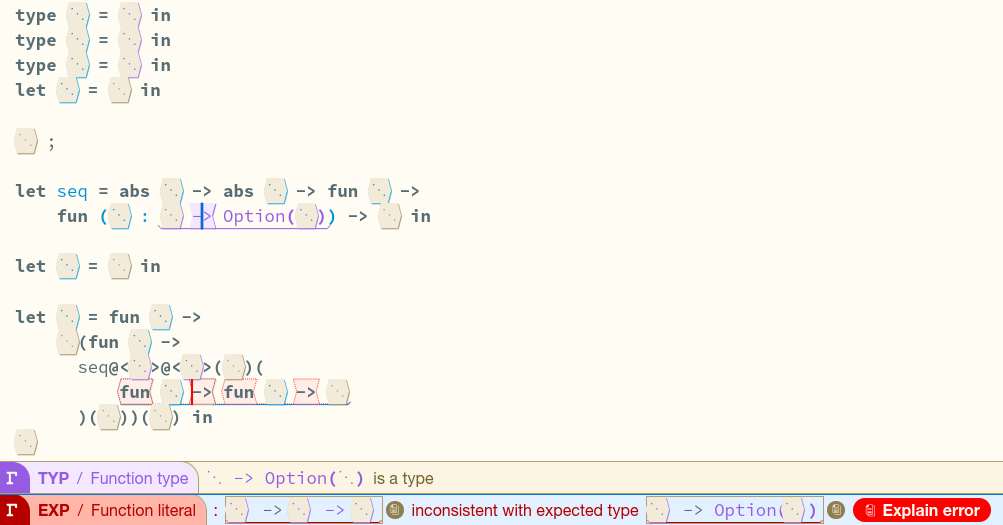}};
\draw[black!35, line width=0.35pt] (img.south west) rectangle (img.north east);
\begin{scope}[x={(img.south east)},y={(img.north west)}]
  \node[purpco] (f1) at (0.55,0.585) {source of the (erroneous) expected type:\\
    \texttt{seq}'s annotation \texttt{A -> Option((String, B))}};
  \draw[purpl] (f1.west) -- (0.2,0.66);
  \node[errco] (f2) at (0.62,0.425) {the (correct) function literals marked\\ as an error in standard localisation};
  \draw[errl] (f2.west) -- (0.3,0.32);
  \node[foldco] (f3) at (0.62,0.257) {everything else\\folds away};
\end{scope}
\end{tikzpicture}

\caption{Explain-error slice for the parser combinator.}
\label{fig:explain-error}
\end{figure}
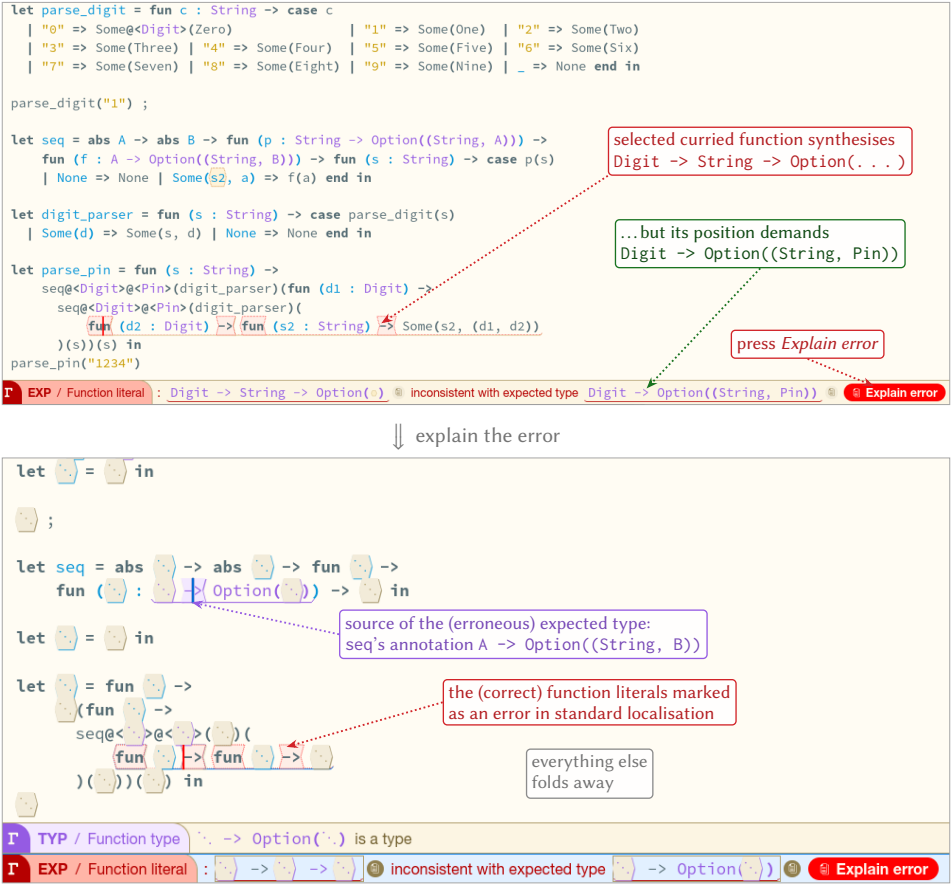

{\Cref{fig:explain-error} considers a much more complex error where a sequencing parser combinator has been defined incorrectly, but used correctly. As is usual, the error has been incorrectly localised to the usage site, trusting the (incorrect) annotation. Specifically, \texttt{seq}'s continuation is annotated as if it returns an \texttt{Option} immediately, rather than first accepting the remaining \texttt{String} threaded out of the previous parser. The use site provides the intended parser continuation.

\emph{Explain error} queries the slices on only the outermost inconsistency in the types (the function vs the option in the return type); the resulting slice collapses the program to just the function shape of the continuation and the offending annotation.

Crucially, the slice did not assume the annotation is the ground truth. It shows the user both the expectation's \emph{source} (the annotation) and the \emph{usage} (the function), leaving the programmer to judge which is wrong. This is justified as programmers often add annotations incorrectly to existing code (49\% of the time~\cite{TypeAnnotationStudy}).} Should the user instead wish to trust the context unilaterally, they can simply only consider the synthesis slice (omitting the analysis slice entirely). We later make slicing of type errors precise, applying slicing to marked ill-typed programs (\S\ref{sec:marking}).

\clearpage

\section{Language Scope}
\label{sec:background}

This section introduces the key characteristics of the languages to which type slicing can be applied.

\subsection{Bidirectional Types}
\label{sec:bidirectional-types}

Type slicing is formalised for \textit{bidirectional type systems}~\cite{BidirectionalTypes}. Such a type system is an algorithmic specification of typing judgements, naturally lending itself to direct typing algorithms which do not require \textit{guessing} of types during type checking. These systems do require annotations, but annotation burden is lessened by locally inferring types~\cite{LocalInference}.

This is achieved by specifying the \textit{mode} of the type parameter in a typing judgement, distinguishing when it is an \textit{input} (type analysis/checking) and when it is an \textit{output} (type synthesis):
\[\synthesis{e}{\tau}\]
Read as: \textit{$e$ synthesises type $\tau$ under typing assumptions $\Gamma$ and type variable assumptions $\Delta$}. The type $\tau$ is an \textit{output}.
\[\analysis{e}{\tau}\]
Read as: \textit{$e$ analyses against type $\tau$ under typing assumptions $\Delta, \Gamma$}. The type $\tau$ is an \textit{input}. This is where the terminology of \textit{synthesis / analysis} slices is derived from.

Such languages have two fundamental rules:
\[\inference[\textsc{SVar}]{x : \tau \in \Gamma}{\synthesis{x}{\tau}} \qquad
  \inference[\textsc{Sub}]{\synthesis{e}{\tau}}{\analysis{e}{\tau}}\]
Variables synthesise their type from the typing assumptions. Subsumption bridges the two modes: any term that can synthesise a type also analyses against that same type.

\subsection{Downwards Static Graduality}
\label{sec:gradual-types}

To calculate a type slice we will need a way to \textit{gradually} omit sub-expressions from the program, and type check these less \textit{`precise'} partial programs at less or equally `precise' types. Formally, precision must form two partial orders (notated $\sqsubseteq$) over types and expressions, which can be lifted onto typing assumptions $\Gamma$.

\begin{definition} A bidirectional calculus satisfies \textit{downwards static graduality} if, for all $\Delta, \Gamma', \Gamma, e', e, \tau$ such that $\Gamma' \sqsubseteq \Gamma$ and $e' \sqsubseteq e$, both:
\begin{itemize}
  \item If $\synthesis{e}{\tau}$ then there exists $\tau' \sqsubseteq \tau$ such that $\synthesis[\Delta;\ \Gamma']{e'}{\tau'}$.
\item For all $\tau' \sqsubseteq \tau$, if $\analysis{e}{\tau}$ then $\analysis[\Delta; \Gamma']{e'}{\tau'}$.
\end{itemize}
\end{definition}

For example, this precision order can correspond to the precision relation on types in a gradually typed system \cite{GradualFunctional}. In this situation this graduality theorem corresponds to the downwards static case\footnote{Hence the name.} of the graduality theorem in the refined criteria for gradual types \cite{GradualRefined}, adapted to a bidirectional setting. The calculus we demonstrate type slicing on in this paper goes even further by allowing maximally incomplete programs, omitting any program sub-expression, based upon the Hazel calculus~\cite{HazelLivePaper}.

Although not necessary, it is also useful that precision forms a lattice with meets and joins to perform optimisations to type slicing calculation (\S\ref{sec:algorithms}).

\section{The Core Calculus}
\label{sec:core-calculus}

In this section, we present the core calculus which type slicing builds upon. It is a bidirectional, gradually typed lambda calculus, with precision on types and terms, based on the Hazel calculus~\cite{HazelLivePaper}, extended with products, sums, and explicit System~F polymorphism. We prove that this calculus satisfies graduality (\S\ref{sec:gradual-types}), and we define and prove several lattice structures on precision.

\subsection{Syntax \& Relations}
\label{sec:core-syntax}

The syntax of the core calculus is given in Figure~\ref{fig:syntax}. $1$ is the unit type, and $\gap$ is a gap / hole serving a dual role: in the gradual typing interpretation it is the dynamic type, and in the slicing interpretation it represents unrequired type information (to explain a query). Similarly a $\gap$ in expressions represents omitted sub-expressions, irrelevant to the explanation resulting from a query. In type checking, omitted expressions synthesise the dynamic (omitted) type.

\begin{figure}[H]
\fbox{Syntax}\ \ \ Types $\tau$ and expressions $e$
\begin{align*}
\tau &::= \gap \mid 1 \mid \tau \funarr \tau \mid \tau \times \tau \mid \tau + \tau \mid \forall\alpha.\;\tau \mid \alpha \\[4pt]
e &::= \gap \mid () \mid x \mid \lambda x \mathbin{:} \tau.\; e \mid \lambda x.\; e \mid e(e) \mid (e,\; e) \mid \pi_1\; e \mid \pi_2\; e \\
  &\quad\mid\; \iota_1\; e \mid \iota_2\; e \mid (\mathbf{case}\;e\;\mathbf{of}\;x.\;e \mid y.\;e) \mid \Lambda\alpha.\; e \mid e\langle\tau\rangle \mid \mathbf{let}\;x = e\;\mathbf{in}\;e
\end{align*}

\caption{Syntax of the core calculus up to $\alpha$-equivalence (on variables $x$).}
\label{fig:syntax}
\end{figure}

Base types such as booleans are omitted from the core syntax, being encoded using sums in the usual way (e.g.\ $\mathit{Bool} \triangleq 1 + 1$); we use them freely in examples.

\subsubsection{Consistency}

Type consistency (Figure~\ref{fig:consistency}) is exactly as presented in standard gradual type systems (\S\ref{sec:gradual-types}). That is, a weakening of equality, where every type is consistent with $\gap$, and compound types are consistent when their components are. Consistency is reflexive and symmetric but \textit{not} transitive.

\begin{figure}[H]
\fbox{$\tau_1 \sim \tau_2$}\ \ \ $\tau_1$ is consistent with $\tau_2$
\[\inference[$\sim\gap_1$]{}{\tau \sim \gap} \qquad
  \inference[$\sim\gap_2$]{}{\gap \sim \tau} \qquad
  \inference[$\sim{1}$]{}{1 \sim 1} \qquad
  \inference[$\sim{\funarr}$]{\tau_1 \sim \tau_1' & \tau_2 \sim \tau_2'}{\tau_1 \funarr \tau_2 \sim \tau_1' \funarr \tau_2'}\]
\[\inference[$\sim{\times}$]{\tau_1 \sim \tau_1' & \tau_2 \sim \tau_2'}{\tau_1 \times \tau_2 \sim \tau_1' \times \tau_2'} \qquad
  \inference[$\sim{+}$]{\tau_1 \sim \tau_1' & \tau_2 \sim \tau_2'}{\tau_1 + \tau_2 \sim \tau_1' + \tau_2'} \qquad
  \inference[$\sim{\forall}$]{\tau \sim \tau'}{\forall\alpha.\;\tau \sim \forall\alpha.\;\tau'}\]
\caption{Type consistency.}
\label{fig:consistency}
\end{figure}

\subsubsection{Precision}

Precision (Figure~\ref{fig:precision}) is a partial order that organises terms by \textit{information content}: for types, $\tau' \sqsubseteq \tau$ means $\tau$ has at least as much type information as $\tau'$. Here, $\gap$ is the global minimum. Expression precision is defined analogously; for instance, the slices of $\lambda x \mathbin{:} 1.\; x$ form the precision chain $\gap \sqsubseteq \lambda x \mathbin{:} \gap.\; \gap \sqsubseteq \lambda x \mathbin{:} 1.\; \gap \sqsubseteq \lambda x \mathbin{:} 1.\; x$, each step filling in one more part of the term.

\begin{figure}[H]
\fbox{$\tau \sqsubseteq \tau'$}\ \ \ $\tau$ is less precise than $\tau'$
\[\inference[$\sqsubseteq\gap$]{}{\gap \sqsubseteq \tau} \qquad
  \inference[$\sqsubseteq{1}$]{}{1 \sqsubseteq 1} \qquad
  \inference[$\sqsubseteq{\alpha}$]{}{\alpha \sqsubseteq \alpha} \qquad
  \inference[$\sqsubseteq{\funarr}$]{\tau_1 \sqsubseteq \tau_1' & \tau_2 \sqsubseteq \tau_2'}{\tau_1 \funarr \tau_2 \sqsubseteq \tau_1' \funarr \tau_2'}\]
\[\inference[$\sqsubseteq{\times}$]{\tau_1 \sqsubseteq \tau_1' & \tau_2 \sqsubseteq \tau_2'}{\tau_1 \times \tau_2 \sqsubseteq \tau_1' \times \tau_2'} \qquad
  \inference[$\sqsubseteq{+}$]{\tau_1 \sqsubseteq \tau_1' & \tau_2 \sqsubseteq \tau_2'}{\tau_1 + \tau_2 \sqsubseteq \tau_1' + \tau_2'} \qquad
  \inference[$\sqsubseteq{\forall}$]{\tau \sqsubseteq \tau'}{\forall\alpha.\;\tau \sqsubseteq \forall\alpha.\;\tau'}\]
\caption{Type precision.}
\label{fig:precision}
\end{figure}

\paragraph{Precision on Typing Assumptions.} In this calculus, we represent typing assumptions as \textit{finite} partial functions from variable names to types, with domain $\mathrm{dom}(\Gamma)$. Precision is defined by domain inclusion and pointwise precision on types in the (smaller) domain:
\[\gamma_1 \sqsubseteq \gamma_2 \iff \mathrm{dom}(\gamma_1) \subseteq \mathrm{dom}(\gamma_2) \text{ and } \gamma_1(x) \sqsubseteq \gamma_2(x) \text{ for all } x \in \mathrm{dom}(\gamma_1)\]
That is, a less precise type assumptions set either assumes fewer variables' types, or assumes less precise types on the present variables.

\begin{remark}\label{rem:assms-mechanisation}
Typical literature refers to these \textit{typing assumptions} as `typing contexts'. We avoid this terminology as we use the \textit{context} to refer to the syntactic context around a subterm, i.e. in the same sense as used in contextual dynamics (\textcite[Chapter~5]{Harper2016}).
\end{remark}

\subsection{Lattice Properties}
\label{sec:core-lattice}

\subsubsection{Meets and Joins}
\label{sec:core-meets-joins}

The \textit{meet} $\tau_1 \sqcap \tau_2$ is the greatest lower bound (GLB) under precision, retaining only sub-terms present in both terms. The \textit{join} $\tau_1 \sqcup \tau_2$ is the least upper bound (LUB) under precision, omitting only subterms omitted in both components. For example:
\[(1 \funarr \gap) \sqcap (\gap \funarr 1) \;=\; \gap \funarr \gap, \qquad\quad (1 \funarr \gap) \sqcup (\gap \funarr 1) \;=\; 1 \funarr 1.\]
Meets are \textit{total}, but joins are defined only between \textit{consistent} types.

\subsubsection{The Slice Lattice}
\label{sec:slice-lattice}

For a fixed type $\tau$, the \textit{slice lattice} of $\tau$ is the \textit{lower set} under precision:
\[\lfloor\tau\rfloor \;=\; \{\,\upsilon \mid \upsilon \sqsubseteq \tau\,\}\]
$\lfloor e \rfloor$, $\lfloor \C \rfloor$, and $\lfloor \Gamma \rfloor$ are defined analogously for expressions, contexts, and assumptions. Each slice lattice has top $\top = \tau$ and bottom $\bot = \gap$. Then precision, meets, and joins simply lift unchanged into the slice lattice. Importantly all terms in the slice lattice are consistent, so joins are always defined:

\begin{lemma}[Slices are Mutually Consistent]\label{lem:slices-consistent}
If\/ $\upsilon_1 \sqsubseteq \tau$ and $\upsilon_2 \sqsubseteq \tau$, then $\upsilon_1 \sim \upsilon_2$.
\end{lemma}

Hence, all terms in the slice lattice are consistent, making it a well-defined \textit{bounded lattice}. The meet and join are also distributive. Further, slice lattices are trivially finite:

\begin{theorem}\label{thm:slice-finite}
Each slice lattice $\lfloor\tau\rfloor$, $\lfloor e \rfloor$, $\lfloor \C \rfloor$, and $\lfloor \Gamma \rfloor$ is finite.
\end{theorem}

\begin{theorem}\label{thm:slice-lattices}
Each slice lattice $\lfloor\tau\rfloor$, $\lfloor e \rfloor$, $\lfloor \C \rfloor$, and $\lfloor \Gamma \rfloor$ is a bounded distributive lattice.
\end{theorem}
Moreover, being finite and distributive, each slice lattice is bi-Heyting~\cite{DaveyPriestley}. We will later introduce and use co-Heyting subtraction to represent `disjoint' types (\S\ref{sec:algorithms}).

\paragraph{Notation: Slice Lattices as Highlighted Terms.} An element of a slice lattice can be notated as a highlight of the original (top) term with omitted regions (the \gap{}s) left unhighlighted. For instance, the slice $\mathit{Int} \funarr \gap$ in $\lfloor\mathit{Int} \funarr \mathit{Bool}\rfloor$ can be represented as $\sli{\mathit{Int} \funarr}\, \mathit{Bool}$.

\subsection{Typing Rules}
\label{sec:core-typing}

The bidirectional typing judgements are mostly standard; the complete synthesis and analysis rule sets appear in the supplementary material, and Figure~\ref{fig:typing-rules-excerpt} presents some of the major rules. $\wf{\tau}$ denotes when a type $\tau$ is \textit{well-formed} under available type variables $\Delta$.

\begin{figure}[h!t]
\small
\[\inference[\synr{Var}]{x : \tau \in \Gamma}{\synthesis{x}{\tau}} \qquad
  \inference[\synr{$\lambda$Ann}]{\wf{\tau_1} & \synthesis[\Delta;\,\Gamma, x \mathbin{:} \tau_1]{e}{\tau_2}}{\synthesis{\lambda x \mathbin{:} \tau_1.\; e}{\tau_1 \funarr \tau_2}}\]
\[\inference[\synr{let}]{\synthesis{e_1}{\tau_1} & \synthesis[\Delta;\,\Gamma, x \mathbin{:} \tau_1]{e_2}{\tau_2}}{\synthesis{\mathbf{let}\;x = e_1\;\mathbf{in}\;e_2}{\tau_2}} \qquad
  \inference[\synr{$\Lambda$}]{\synthesis[\Delta, \alpha;\,\Gamma]{e}{\tau}}{\synthesis{\Lambda\alpha.\; e}{\forall\alpha.\;\tau}}\]
\[\inference[\synr{App}]{\synthesis{e_1}{\tau} & \tau \sqcup (\gap \funarr \gap) \equiv \tau_1 \funarr \tau_2 & \analysis{e_2}{\tau_1}}{\synthesis{e_1(e_2)}{\tau_2}}\]
\[\inference[\synr{TApp}]{\wf{\sigma} & \synthesis{e}{\tau} & \tau \sqcup \forall\alpha.\;\gap \equiv \forall\alpha.\;\tau'}{\synthesis{e\langle\sigma\rangle}{[\alpha \mapsto \sigma]\,\tau'}}\]
\[\inference[\synr{case}]{\synthesis{e}{\tau} & \tau \sqcup (\gap + \gap) \equiv \tau_1 + \tau_2 & \tau_1' \sim \tau_2' \\ \synthesis[\Delta;\,\Gamma, x \mathbin{:} \tau_1]{e_1}{\tau_1'} & \synthesis[\Delta;\,\Gamma, y \mathbin{:} \tau_2]{e_2}{\tau_2'}}{\synthesis{\mathbf{case}\;e\;\mathbf{of}\;x.\;e_1 \mid y.\;e_2}{\tau_1' \sqcup \tau_2'}}\]
\caption{Selected synthesis rules.}
\label{fig:typing-rules-excerpt}
\end{figure}

We use joins to replicate type matching in the gradual typing literature.\footnote{Typically notated $\tau \blacktriangleright \tau_1 \funarr \tau_2$.} $\tau \equiv \tau_1 \funarr \tau_2$ extracts the domain and codomain of $\tau$, or the gap type $\gap$ if $\tau \equiv \gap$. Case expressions can synthesise a type when their branches synthesise \textit{consistent} types. There is a dedicated rule for unannotated let bindings, which differ from function applications in a bidirectional setting due to type information flowing from the binding into the body. The typing rules satisfy graduality (\S\ref{sec:gradual-types}).
\begin{theorem}\label{thm:graduality-syn}
The synthesis judgement $\synthesis[\_]{\_}{\_}$ and analysis judgement $\analysis[\_]{\_}{\_}$ satisfy downwards static graduality.
\end{theorem}

\section{Synthesis Slices}
\label{sec:synthesis}

\providecommand{\synslice}[2]{\mathit{SynSlice}\;#1 \mathbin{\blacktriangleleft} #2}
\providecommand{\minsynslice}[2]{\mathit{MinSynSlice}\;#1 \mathbin{\blacktriangleleft} #2}
\providecommand{\exactsynslice}[2]{\mathit{ExactSynSlice}\;#1 \mathbin{\blacktriangleleft} #2}
\providecommand{\isMin}[1]{\mathsf{IsMinimal}\,(#1)}
\providecommand{\joinsyn}{\sqcup}
\providecommand{\meetsyn}{\sqcap}
\providecommand{\pairsyn}{\times}

\noindent The type information of a subexpression in a bidirectional type system is derived either from \textit{within} the term by synthesis, or from the surrounding context. This section defines \textit{synthesis slices} which concern how to explain an expression's \textit{synthesised} type, that is, the type information \textit{internal} to the expression.

\subsection{Synthesis Slices}
\label{sec:syn-slice-def}

Consider a `program' $P = (\Gamma, e)$, that is, typing assumptions and an expression. When the program synthesises a type (in some type variable context $\Delta$)\footnote{Type variable contexts are not sliced as the only information such slices would encode is whether a type variable was used, which can be extracted purely syntactically from the resulting slice. Typing assumptions are different because we can \textit{partially} require an assumption.}, all of its type information is contained within the synthesis derivation $D \;:\; \synthesis{e}{\tau}$.

Given the derivation $D$, we may \textit{query} the type $\tau$ by taking a slice $\upsilon$ in $\lfloor \tau \rfloor$. A synthesis slice is a `program slice' $\rho = (\gamma, \varsigma)$ in $\lfloor \Gamma , e\rfloor$ which synthesises a type at least as precise as the query $\upsilon$; this corresponds to a highlighted version of the original program. 

Intuitively, $\upsilon$ specifies the part of $\tau$ we wish to explain; for example, querying $\upsilon = \mathit{Int} \funarr \gap$ for $\tau = \mathit{Int} \funarr \mathit{Bool}$ asks ``which parts of the program account for the domain being $\mathit{Int}$?''. The slice might not include \textit{all} the program regions relevant to the query type, but will necessarily provide a consistent subset which suffices to totally \textit{explain} (independently synthesise) the query.

\noindent
\begin{minipage}[c]{0.56\linewidth}
\begin{definition}[Synthesis slice]
\label{def:syn-slice}
A \textit{synthesis slice} of $D : \synthesis{e}{\tau}$ on $\upsilon \in \lfloor \tau \rfloor$, written $\synslice{D}{\upsilon}$, is a pair $(\gamma, \varsigma) \in \lfloor \Gamma, e\rfloor$ such that $\Delta;\gamma \vdash \varsigma\;\synmode\;\phi$ for some $\phi \sqsupseteq \upsilon$ (with $\phi \in \lfloor \tau \rfloor$ by graduality, \cref{thm:graduality-syn}).
\end{definition}
\end{minipage}\hfill
\begin{minipage}[c]{0.40\linewidth}
\centering
\vspace{6pt}
\begin{tikzcd}[column sep=large, row sep=normal]
& \rho \arrow[r, "\sqsubseteq" {sloped}] \arrow[d, "\synmode"'] & (\Gamma, e) \arrow[d, "\synmode"] \\
\upsilon \arrow[r, "\sqsubseteq" {sloped}] & \phi \arrow[r, induced, "\sqsubseteq" {sloped}, "\text{(by graduality)}"', dashed] & \tau
\end{tikzcd}
\end{minipage}

\medskip

Many such synthesis slices exist. In particular, the original program $(\Gamma, e)$ is necessarily a synthesis slice; and for $\synslice{D}{\gap}$, every slice of $(\Gamma, e)$ is a valid synthesis slice. We refer to the $\upsilon \sqsubseteq \phi$ condition as \textit{validity}. It is intentional that this is \textit{not} an exact equality, as demonstrated by the counterexample below:

\begin{counterexamplebox}

\begin{counterexample}[Non-Existence of all Exact Synthesis Slices]\label{lem:exact-not-exist}\leavevmode\\
Consider $x \mathbin{:} \mathbf{1} \funarr \mathbf{1} \vdash (x, x) \synmode (1 \funarr 1) \times (1 \funarr 1)$, and query $\upsilon = (\mathbf{1} \funarr \gap) \times (\gap \funarr \mathbf{1})$.

Any strictly smaller slice of either $x : 1 \funarr 1$ or $(x, x)$ synthesises a type strictly less precise than $\upsilon$. Yet the original program synthesises $(1 \funarr 1) \times (1 \funarr 1)$, strictly more precise than $\upsilon$.
\end{counterexample}
\end{counterexamplebox}

\subsection{Minimality}
\label{sec:minimality}

Synthesis slices provide \textit{some} explanation of the query, but many are clearly useless, such as the original program itself, which corresponds to fully highlighting the program. We want to find minimal slices: those where further slicing any part of the program makes the slice invalid (synthesising a type insufficient to cover the query).
\begin{definition}[Minimality]
\label{def:isminimal}
For an element $a$ of a poset, $a$ is minimal, $\isMin{a}$, when every $a'$ with $a' \sqsubseteq a$ is in fact equal to $a$.
\end{definition}
\begin{definition}[Minimal Synthesis Slices]
A $\minsynslice{D}{\upsilon}$ is a $\synslice{D}{\upsilon}$ minimal in $\synslice{D}{\upsilon}$ under program-slice precision on $\lfloor \Gamma, e\rfloor$.
\end{definition}

\subsection{Existence of Minimal Slices}
\label{sec:min-exists}

Crucially, any synthesis slice has a minimal slice below it (or is minimal):
\begin{majortheorembox}

\begin{theorem}[Existence of minimal slices]\label{thm:min-exists}
For every $s : \synslice{D}{\upsilon}$, there exists $m : \minsynslice{D}{\upsilon}$ with $m \sqsubseteq s$.
\end{theorem}
\end{majortheorembox}

\begin{proof}
The slice lattice $\lfloor \Gamma, e \rfloor$ is finite. Hence, $\{\,s' : \synslice{D}{\upsilon} \mid s' \sqsubset s\,\}$ is finite and computable by enumeration. If this set is empty, then $s$ is minimal; otherwise, recurse on any of its elements. Each step strictly decreases precision, so finiteness ensures termination.
\end{proof}

Consequently, any derivation $D$ has a minimal synthesis slice, as the original program is a trivial inhabitant of $\synslice{D}{\upsilon}$ for all $\upsilon$.
\paragraph{A Brute-Force Algorithm via Graduality} The proof above is not a practical algorithm. In practice, for any term, the list of maximal slices strictly less precise than a given term can be computed efficiently by omitting exactly \textit{one} leaf from the AST.

The algorithm picks the first maximal strict slice that satisfies the query and recurses; if none of the slices satisfy the query, then a minimal slice has been found, by graduality (\cref{thm:graduality-syn}). This is $\mathcal{O}(n \times T)$ where $n$ is the program size and $T$ the type checking complexity; type checking is linear in the program size for this calculus, so finding a minimal slice is $\mathcal{O}(n^2)$.

See Figure~\ref{fig:syn-slice-existence} for an example lattice to descend through, stopping just before descending below the given query; all three minima are marked.

\tikzset{
  slprog/.style={draw=latticecolor!55, fill=white, rounded corners=2pt,
                 inner sep=2.2pt, font=\scriptsize},
  slmin/.style={draw=OliveGreen!70!black, fill=white, line width=0.9pt,
                ellipse, minimum height=0.95cm, inner xsep=2pt, inner ysep=2pt, font=\scriptsize},
  slexact/.style={draw=OliveGreen!75!black, fill=white, line width=1.2pt,
                  rectangle, minimum width=0.85cm, minimum height=0.85cm, inner sep=2pt, font=\scriptsize},
  slhi/.style={draw=RoyalBlue!85!black, fill=white, line width=1.3pt, rounded corners=2pt,
               inner sep=2.2pt, font=\scriptsize},
  slprogarr/.style={->, latticecolor!75, line width=0.95pt},
  slsyn/.style={->, BrickRed, line width=0.6pt, dashed,
                dash pattern=on 3pt off 2pt, shorten >=2.5pt, shorten <=2.5pt},
  slpath/.style={RoyalBlue, opacity=0.5, line width=6.5pt,
                 line cap=round, line join=round},
  slfront/.style={OliveGreen!55!black, line width=1.5pt, line cap=round},
  sldot/.style={OliveGreen!55!black, line width=0.5pt, dotted},
}
\newcommand{\sliceNodes}{
  \node[slprog] (T)   at ( 0.0, 9.0) {$(P,\;\tau)$};
  \node[slprog] (A)   at (-2.7, 7.4) {$(\rho_a,\;\phi_a)$};
  \node[slprog] (B)   at ( 0.1, 7.3) {$(\rho_b,\;\phi_b)$};
  \node[slprog] (C)   at ( 2.7, 7.4) {$(\rho_c,\;\phi_c)$};
  \node[slprog] (E)   at (-1.9, 5.9) {$(\rho_e,\;\phi_e)$};
  \node[slprog] (F)   at ( 1.4, 6.0) {$(\rho_f,\;\phi_f)$};
  \node (G)   at ( 3.0, 5.0) {};
  \node[slprog] (H)   at (-3.0, 4.4) {$(\rho_h,\;\phi_h)$};
  \node[slprog] (I)   at (-0.5, 4.2) {$(\rho_i,\;\upsilon)$};
  \node[slprog] (J)   at ( 2.5, 3.7) {$(\rho_j,\;\phi_j)$};
  \node[slprog] (K)   at (-1.7, 2.5) {$(\rho_k,\;\phi_k)$};
  \node[slprog] (Bot) at ( 0.7, 1.1) {$(\gap,\;\gap)$};
}
\newcommand{\sliceEdges}{
  \foreach \p/\c in {T/A,T/B,T/C,A/E,A/H,B/E,B/F,C/F,E/I,F/I,F/J,H/K,I/K,I/Bot,J/Bot,K/Bot}
    \draw[slprogarr] (\p) -- (\c);
}
\newcommand{\sliceFrontier}{
  \draw[slfront] (-3.9,3.5) .. controls (-2.4,3.5) and (-1.1,4.7) .. (3.9,4.3);
}

\begin{figure}[t]
\centering
\begin{subfigure}[t]{0.49\linewidth}
\centering
\begin{tikzpicture}[scale=0.68, transform shape, >={Stealth[length=4pt,width=3pt]}]
\begin{scope}[on background layer]
  \fill[OliveGreen!14] (-3.9,9.8) -- (-3.9,3.5)
    .. controls (-2.4,3.5) and (-1.1,4.7) .. (3.9,4.3) -- (3.9,9.8) -- cycle;
  \sliceFrontier
  \draw[sldot] (-3.9,9.8) -- (3.9,9.8)  (-3.9,9.8) -- (-3.9,3.5)  (3.9,9.8) -- (3.9,4.3);
\end{scope}
\sliceNodes
\sliceEdges
\draw[slprogarr] (G) -- (J);
\draw[slprogarr] (C) -- (G);
\foreach \p/\c in {T/A,T/B,T/C,A/E,A/H,B/E,B/F,C/F,C/G,E/I,F/I,F/J,G/J,H/K,I/K,I/Bot,J/Bot,K/Bot}
  \draw[slsyn] (\p) to[bend right=7] (\c);
\draw[slsyn] (F) -- (G);
\draw[slsyn] (E) -- (H);
\node[slmin] at (G) {$(\rho_g,\;\phi_g)$};
\node[slmin] at (H) {$(\rho_h,\;\phi_h)$};
\node[slexact] at (I) {$(\rho_i,\;\upsilon)$};
\node[anchor=north west, font=\scriptsize\bfseries, OliveGreen!35!black] at (-3.85,9.6) {$\upsilon \sqsubseteq \phi_i$};
\node[anchor=south west, font=\scriptsize\bfseries, gray!50!black] at (-3.85,0.7) {$\upsilon \not\sqsubseteq \phi_i$};
\end{tikzpicture}
\caption{Synthesis slices above query $\upsilon$ in the shaded region, with 3 minima consisting of 2 inexact minima (ellipses) and one minimum (square) synthesising exactly $\upsilon$.}
\label{fig:syn-slice-existence}
\end{subfigure}\hfill
\begin{subfigure}[t]{0.49\linewidth}
\centering
\begin{tikzpicture}[scale=0.68, transform shape, >={Stealth[length=4pt,width=3pt]}]
\begin{scope}[on background layer]
  \fill[OliveGreen!8]  (-3.9,9.8) -- (-3.9,2.2) .. controls (-1.2,1.8) and (1.2,1.8) .. (3.9,2.2) -- (3.9,9.8) -- cycle;
  \fill[OliveGreen!18] (-3.9,9.8) -- (-3.9,3.5) .. controls (-2.4,3.5) and (-1.1,4.7) .. (3.9,4.3) -- (3.9,9.8) -- cycle;
  \fill[OliveGreen!32] (-3.9,9.8) -- (-3.9,6.6) .. controls (-1.5,6.9) and (1.5,7.0) .. (3.9,6.7) -- (3.9,9.8) -- cycle;
  \draw[slfront] (-3.9,6.6) .. controls (-1.5,6.9) and (1.5,7.0) .. (3.9,6.7);
  \sliceFrontier
  \draw[slfront] (-3.9,2.2) .. controls (-1.2,1.8) and (1.2,1.8) .. (3.9,2.2);
  \draw[sldot] (-3.9,9.8) -- (3.9,9.8)  (-3.9,9.8) -- (-3.9,2.2)  (3.9,9.8) -- (3.9,2.2);
\end{scope}
\sliceNodes
\sliceEdges
\foreach \p/\c in {T/A,T/B,T/C,A/E,A/H,B/E,B/F,C/F,E/I,F/I,F/J,H/K,I/K,I/Bot,J/Bot,K/Bot}
  \draw[slsyn] (\p) to[bend right=7] (\c);
\draw[slsyn] (E) -- (H);
\begin{scope}[on background layer]
  \draw[slpath] (T) -- (B) -- (E) -- (I) -- (K) -- (Bot);
\end{scope}
\node[slhi] at (B) {$(\rho_b,\;\phi_b)$};
\node[slhi] at (I) {$(\rho_i,\;\upsilon_1)$};
\node[slhi] at (K) {$(\rho_k,\;\phi_k)$};
\node[anchor=west, font=\normalsize\bfseries, RoyalBlue!80!black] at ([xshift=3pt]B.east) {$m_2$};
\node[anchor=east, font=\normalsize\bfseries, RoyalBlue!80!black] at ([xshift=-3pt]I.west) {$m_1$};
\node[anchor=east, font=\normalsize\bfseries, RoyalBlue!80!black] at ([xshift=-3pt]K.west) {$m_0$};
\node[anchor=west, font=\normalsize, OliveGreen!40!black] at (3.95,6.7) {$\upsilon_2$};
\node[anchor=west, font=\normalsize, OliveGreen!40!black] at (3.95,4.3) {$\upsilon_1$};
\node[anchor=west, font=\normalsize, OliveGreen!40!black] at (3.95,2.2) {$\upsilon_0$};
\end{tikzpicture}
\caption{Monotonically descending minima $m_2, m_1, m_0$ for query refinements $\upsilon_2 \to \upsilon_1 \to \upsilon_0$, calculated by descending along the \textcolor{RoyalBlue!85!black}{highlighted} path.}
\label{fig:syn-slice-refinement}
\end{subfigure}
\caption{Two similar Hasse diagrams of program slices $\rho_i \in \lfloor P \rfloor$, ordered by {\color{latticecolor}strict program precision $\sqsubset$} (solid arrows), paired with corresponding synthesised types $\phi_i \in \lfloor \tau \rfloor$. Overlaid with precision relations on types ({\color{BrickRed}dashed arrows}).}
\label{fig:syn-slices}
\end{figure}
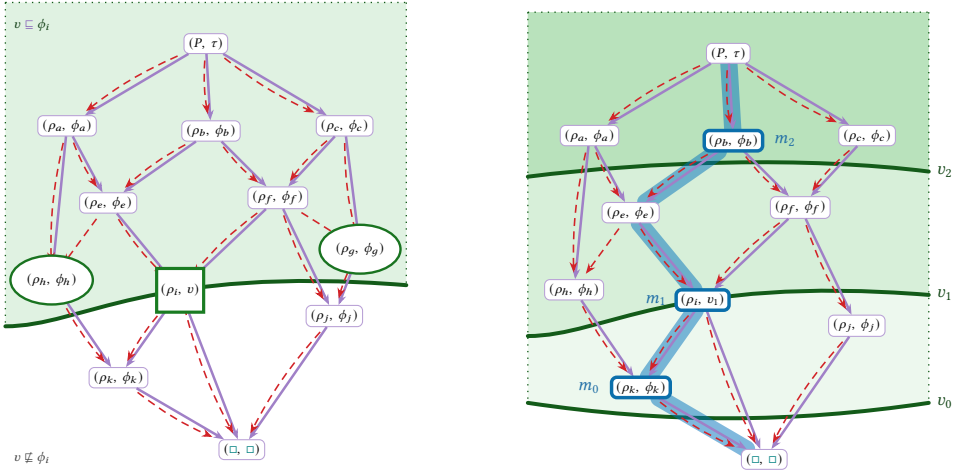

\subsection{Refined Slices}
\label{sec:sub-slices}
In practice, the user may look at a large query $\upsilon$ getting a large program slice $\rho$. Therefore they may want to refine the query to focus on a subpart, $\upsilon' \sqsubseteq \upsilon$, that they do not quite understand. It is always possible to pick a consistent program slice $\rho' \sqsubseteq \rho$ which is a refinement of the previous slice by \textit{monotonicity} (which follows from graduality, \cref{thm:graduality-syn}):

\begin{majortheorembox}
\begin{theorem}[Monotonicity of minimal slices]\label{thm:mono}
If $\upsilon_1 \sqsubseteq \upsilon_2$ in $\lfloor\tau\rfloor$ and $m_2 : \minsynslice{D}{\upsilon_2}$, then there exists $m_1 : \minsynslice{D}{\upsilon_1}$ with $m_1 \sqsubseteq m_2$.
\end{theorem}
\end{majortheorembox}

This is desirable because the user will not be shown \textit{additional} information after a query refinement which would require further reasoning by the user to understand. \Cref{fig:syn-slice-refinement} illustrates a monotonic refinement path along 3 refined queries on the example lattice. If we did not follow a monotonic path during refinement we would get `branch hopping':

\paragraph{Branch hopping.} This is where highlighted regions switch between branches when refining the query, requiring the user to switch their reasoning between the contents in each branch. For example, refining $\upsilon = \mathbf{1}\funarr\mathbf{1}$ to $\upsilon' = \gap\funarr\mathbf{1}$:
\begin{center}
\begin{tikzpicture}[baseline=(t.base), every node/.style={inner sep=3pt}]
  \node[draw, line width=0.9pt, rounded corners,
        label={[font=\scriptsize\itshape]above:assumptions}] (a) {$x : \sli{\mathbf{1} \funarr \mathbf{1}}$};
  \node[right=6pt of a] (t) {$\sli{\mathbf{if}}\;\mathbf{true}\;\sli{\mathbf{then}\;x\;\mathbf{else}}\;y$};
\end{tikzpicture}
\quad$\xrightarrow{\upsilon' \sqsubseteq \upsilon}$\quad
\begin{tikzpicture}[baseline=(t.base), every node/.style={inner sep=3pt}]
  \node[draw, line width=0.9pt, rounded corners,
        label={[font=\scriptsize\itshape]above:assumptions}] (a) {$y : \mathbf{1}\,\sli{\funarr \mathbf{1}}$};
  \node[right=6pt of a] (t) {$\sli{\mathbf{if}}\;\mathbf{true}\;\sli{\mathbf{then}}\;x\;\sli{\mathbf{else}\;y}$};
\end{tikzpicture}
\end{center}

On the other hand, monotonicity does not necessarily hold in the upwards direction. Consider the $\rho_j$ node in \Cref{fig:syn-slice-refinement}, which is minimal under query $\upsilon_0$, but no minimum for query $\upsilon_1$ is strictly more precise than $\rho_j$.

\subsection{Decompositions}
\label{sec:decompositions}
\mastercases

Instead of calculating slices by brute force (\S\ref{sec:min-exists}), it would be ideal if we could compositionally calculate minimal slices of a term from minimal subslices of its subterms. It is indeed possible to compose synthesis slices into larger synthesis slices, and we prove that all minima can be expressed by some composition; this subsection demonstrates two of the cases. We later consider how to recursively choose the compositions (\S\ref{sec:algorithms}).

\subsubsection{Products}
\label{sec:case-pair}

Given slices $(\gamma_1, \varsigma_1) : \synslice{D_1}{\upsilon_1}$ and $(\gamma_2, \varsigma_2) : \synslice{D_2}{\upsilon_2}$ for the two components of a pair $D = \synr{Pair}\,D_1\,D_2$ (with synthesised types $\phi_1, \phi_2$), we form their product slice $(\gamma_1, \varsigma_1) \pairsyn (\gamma_2, \varsigma_2) : \synslice{D}{\upsilon_1 \times \upsilon_2}$ by joining the two assumption slices and pairing the term slices:
\[(\gamma_1, \varsigma_1) \pairsyn (\gamma_2, \varsigma_2) \;\triangleq\; (\gamma_1 \sqcup \gamma_2,\;(\varsigma_1, \varsigma_2))\]
The pair synthesises (possibly more precise)\footnote{Due to $\gamma_1 \sqcup \gamma_2$ being larger than $\gamma_1$ and $\gamma_2$.} types $\phi_1' \times \phi_2' \sqsupseteq \phi_1 \times \phi_2 \sqsupseteq \upsilon_1 \times \upsilon_2$. For this combinator, we get the following decomposition theorem:

\begin{majortheorembox}
\begin{thmcase}[Product Decomposition]\label{thm:min-prod-decomp}
Let $m : \minsynslice{\synr{Pair}\,D_1\,D_2}{\upsilon_1 \times \upsilon_2}$. Then there exist minimal $m_1 : \minsynslice{D_1}{\upsilon_1}$ and $m_2 : \minsynslice{D_2}{\upsilon_2}$ with $m = m_1 \pairsyn m_2$.
\end{thmcase}
\end{majortheorembox}

However, we shall later see that the converse fails: the product of two minimal slices need not be minimal, since the joined assumptions can leave one component's term slice redundant. So one cannot simply recurse through the derivation, pairing sub-slices on $\upsilon_1$ and $\upsilon_2$ to obtain a minimal slice on $\upsilon_1 \times \upsilon_2$. We explore and resolve this problem later (\S\ref{sec:product-revisited}).

\subsubsection{Let Bindings}
\label{sec:let-decomp}

A more subtle case is that of let-bindings, which have different assumptions for their two premises:
\[\inference[\synr{Let}]{
  \synthesis{e_1}{\tau_1}
  &
  \synthesis[\Delta;\,\Gamma, x:\tau_1]{e_2}{\tau_2}
}{\synthesis{\mathbf{let}\;x = e_1\;\mathbf{in}\;e_2}{\tau_2}}\]
A slice of a let-expression can be constructed from a slice $(\gamma_1, \varsigma_1)$ of the definition (synthesising $\phi_1$) and a slice $(\gamma_2, \varsigma_2)$ of the body, constrained with an assumption consistency condition: the body's used assumptions for the bound variable $x$ must be \textit{at most as precise as} the type synthesised by the slice of the definition, $\gamma_2(x) \mathbin{\sqsubseteq} \phi_1$.

Then, upon joining the assumptions for the two slices, we must remove the assumption on $x$ from the body (notated $\gamma_{\setminus x}$), as this is now provided by the definition:
\[(\gamma_1, \varsigma_1) \mathbin{\mathsf{def}} (\gamma_2, \varsigma_2) \triangleq (\gamma_1 \sqcup (\gamma_2)_{\setminus x},\; \mathbf{let}\;x = \varsigma_1\;\mathbf{in}\;\varsigma_2)\]

\begin{majortheorembox}
\begin{thmcase}[Binding Decomposition]\label{thm:min-def-decomp}
Let $m : \minsynslice{\synr{Let}\,D_1\,D_2}{\upsilon}$. Then there exists $\phi_1 \in \lfloor\tau_1\rfloor$ and minimal slices $m_1 : \minsynslice{D_1}{\phi_1}$ (synthesising $\phi_1'$) and $m_2 = (\gamma_2, \varsigma_2) : \minsynslice{D_2}{\upsilon}$ satisfying $\gamma_2(x) \sqsubseteq \phi_1'$, with $m = m_1 \mathbin{\mathsf{def}} m_2$.
\end{thmcase}
\end{majortheorembox}

\subsection{Contribution Slices}
\label{sec:join-slices}

Finally, a minimal slice only provides an explanation; it does not show \textit{all} information \textit{contributing} in any way to a query type. The latter is a useful notion when performing refactors which change the type, where quickly identifying the regions of the program which rely on the old type is crucial.

Such a region can be uniquely expressed as the \textit{join} of all minimal slices for the query. The contribution slice is efficient in the sense that it is a \textit{least} upper bound of actual minimal slices. For this to work, synthesis slices must be closed under joins: the joined slice must also be \textit{valid}:

\vspace{-0.5em}
\begin{majortheorembox}
\noindent
\begin{minipage}[c]{0.60\linewidth}
\begin{theorem}[Join Closure of Synthesis Slices]\label{thm:join-syn}
For $\rho_1 = (\gamma_1, \varsigma_1) : \synslice{D}{\upsilon_1}$ and $\rho_2 = (\gamma_2, \varsigma_2) : \synslice{D}{\upsilon_2}$, then $\rho_1 \joinsyn \rho_2 : \synslice{D}{\upsilon_1 \sqcup \upsilon_2}$, with the joined slice taken pointwise: $(\gamma_1 \sqcup \gamma_2,\; \varsigma_1 \sqcup \varsigma_2)$.
\end{theorem}
\end{minipage}\hfill
\begin{minipage}[c]{0.36\linewidth}
\centering
\begin{tikzcd}[column sep=1.6em, row sep=1.2em]
\rho_1 \arrow[r, "\sqsubseteq"] \arrow[d, "\synmode"'] & {\color{induced} \rho_1 \sqcup \rho_2} \arrow[d, "\synmode"] & \rho_2 \arrow[l, "\sqsupseteq"'] \arrow[d, "\synmode"] \\
\phi_1 & {\color{induced} \phi'} & \phi_2 \\
\upsilon_1 \arrow[u, "\sqsubseteq" {sloped}] \arrow[r, "\sqsubseteq"] & {\color{induced} \upsilon_1 \sqcup \upsilon_2} \arrow[u, induced, "\sqsubseteq"' {sloped}, dashed] & \upsilon_2 \arrow[u, "\sqsubseteq"' {sloped}] \arrow[l, "\sqsupseteq"']
\end{tikzcd}
\end{minipage}
\end{majortheorembox}

\vspace{-12pt}
\section{Context Typing}
\label{sec:context-typing}

Context typing is a novel construction for bidirectional type systems. It is inspired by one-hole contexts (zippers) of algebraic data types~\cite{Huet1997}, decomposing a data structure into a focused element and its surrounding context. We apply the same idea to \textit{typing derivations}: given a derivation of an expression $e$ and focusing on a sub-term $e'$ in $e$, we focus on the subderivation of $e'$ and pair it with a `context typing derivation' of the surrounding context.

A context typing derivation therefore must record the extra assumptions made available at the focus by additional bound variables introduced in the captured context, and the mode of the focus (whether or not the sub-term synthesises a type). These are necessary for type checking the focus (i.e.\ $e'$) directly from the typing derivation of the context.

\subsection{Context Syntax and Slice Lattice}
\label{sec:context-syntax}

Formally we must first define one-hole contexts on terms themselves (Figure~\ref{fig:contexts}). The focus of the context is notated $\cmark$ and is restricted to a single position in an expression. For any context, a term may be placed at the focus to produce a complete term. This is notated by $\C[e]$, `plugging $e$ into context $\C$'.

\begin{figure}[h]
\begin{minipage}[t]{0.54\linewidth}
\fbox{Syntax}\ \ \ One-hole expression contexts $\C$
\begin{align*}
\C &::= \cmark \mid \lambda x \mathbin{:} \tau.\; \C \mid \lambda x.\; \C \mid \C(e) \mid e(\C) \mid \C\langle\tau\rangle \\
  &\quad\mid\; (\C,\; e) \mid (e,\; \C) \mid \iota_1\; \C \mid \iota_2\; \C \mid \pi_1\; \C \mid \pi_2\; \C \\
  &\quad\mid\; (\mathbf{case}\;\C\;\mathbf{of}\;x.\;e \mid y.\;e) \\
  &\quad\mid\; (\mathbf{case}\;e\;\mathbf{of}\;x.\;\C \mid y.\;e) \mid (\mathbf{case}\;e\;\mathbf{of}\;x.\;e \mid y.\;\C) \\
  &\quad\mid\; \Lambda\alpha.\; \C \mid \mathbf{let}\;x = \C\;\mathbf{in}\;e \mid \mathbf{let}\;x = e\;\mathbf{in}\;\C
\end{align*}
\end{minipage}\hfill
\begin{minipage}[t]{0.43\linewidth}
\fbox{$\C[e]$}\ \ \ Plug expression $e$ into the hole $\cmark$
\vspace{-0.5em}
\begin{align*}
\cmark[e] &= e \\
(\lambda x \mathbin{:} \tau.\; \C)[e] &= \lambda x \mathbin{:} \tau.\; \C[e] \\
(\C_1(e_2))[e] &= \C_1[e](e_2)
\end{align*}
\smallskip\noindent(similarly for all other constructors)
\end{minipage}

\medskip
\fbox{$\gap_{\C}$}\ \ \ Purely structural context: replace all sub-terms with $\gap$, preserving structure
\vspace{-0.5em}
\begin{align*}
\gap_{\cmark} &= \cmark &
\gap_{\lambda x \mathbin{:} \tau.\; \C} &= \lambda x \mathbin{:} \gap.\; \gap_{\C} &
\gap_{\C(e)} &= \gap_{\C}(\gap) \quad \text{etc.}
\end{align*}
\caption{One-hole expression contexts. $\cmark$ marks the hole; $\C[e]$ plugs it.}
\label{fig:contexts}
\end{figure}

\label{sec:context-precision}
Precision extends to contexts, $\C' \sqsubseteq \C$, additionally requiring the focus to be in the same structural position. This is required so that plugging a context preserves precision:

\begin{lemma}[Plug Preserves Precision]\label{lem:plug-precision}
If\/ $\C' \sqsubseteq \C$ and $e' \sqsubseteq e$, then $\C'[e'] \sqsubseteq \C[e]$.
\end{lemma}

This lemma allows slicing to independently slice the \textit{external} type information from the context and the \textit{internal} type information of the focused term, whilst ensuring the whole program slice (plugging the two together) remains a valid slice of the original program.

\subsection{The Context Typing Judgement}
\label{sec:context-classification}

Then, we notate context typing by a judgement $\ctxclass{\C}{d}{\Delta';\ \Gamma'}{m}$, read as: ``under assumptions $\Delta, \Gamma$, context $\C$ in a derivation of outer mode $d$ has its focus typed under $\Delta';\,\Gamma'$ in focus mode $m$''.

The \textit{outer derivation mode} $d$ records whether the whole expression $\C[e]$ is in a synthesis derivation producing type $\tau$ ($d = \synmode\;\tau$) or an analysis derivation against $\tau$ ($d = \anamode\;\tau$). The \textit{focus mode} $m$ records the type the focus $e$ must synthesise ($m = \synmode\;\tau'$), or the type it is checked against ($m = \anamode\;\tau'$). Finally, $\Gamma'$ records the \textit{focus assumptions}. Each rule in these derivations is subject to \textit{additional premises} matching the corresponding term typing rules.

\paragraph{Example.} Consider the derivation of $\synthesis{e_1(e_2)}{\tau_2}$ via \synr{App}, where $\synthesis{e_1}{\tau}$, $\tau \sqcup (\gap \funarr \gap) \equiv \tau_1 \funarr \tau_2$, and $\analysis{e_2}{\tau_1}$.
\begin{itemize}
\item \textbf{Function position} ($e_1$ is the focus, $\C = \cmark(e_2)$): The corresponding context typing is: $\ctxclass{\cmark(e_2)}{\synmode\;\tau_2}{\Delta;\,\Gamma}{\synmode\;\tau}$. The focus must synthesise type $\tau$ as per the focus mode, with the extra premises being $\tau \sqcup (\gap \funarr \gap) \equiv \tau_1 \funarr \tau_2$ and $\analysis{e_2}{\tau_1}$.
\item \textbf{Argument position} ($e_2$ is the focus, $\C = e_1(\cmark)$): The corresponding context typing is: $\ctxclass{e_1(\cmark)}{\synmode\;\tau_2}{\Delta;\,\Gamma}{\anamode\;\tau_1}$. The focus analyses against $\tau_1$ as per the focus mode, with the extra premises being: $\synthesis{e_1}{\tau}$ and $\tau \sqcup (\gap \funarr \gap) \equiv \tau_1 \funarr \tau_2$.
\end{itemize}

Specifying these rules is systematic, one for each typing rule and focus location. We give a few representative rules in Figure~\ref{fig:context-classification}. The only additional rules are for identity contexts $\texttt{s}\cmark$ and $\texttt{a}\cmark$ (the base cases) where the focus mode must match the outer derivation mode.

\begin{figure}[h]
\small
\fbox{$\ctxclass{\C}{d}{\Delta';\,\Gamma'}{m}$}\ \ \ Context $\C$ at outer derivation mode $d$ types its focus under $\Delta';\,\Gamma'$ in focus mode $m$

\bigskip
\fbox{Base cases}\ \ \ identity contexts (focus mode matches outer derivation mode):
\[\inference[\texttt{s$\cmark$}]{}{\ctxclass{\cmark}{\synmode\;\tau}{\Delta;\,\Gamma}{\synmode\;\tau}} \qquad
  \inference[\texttt{a$\cmark$}]{}{\ctxclass{\cmark}{\anamode\;\tau}{\Delta;\,\Gamma}{\anamode\;\tau}}\]

\medskip
\fbox{Synthesis mode: $d = \synmode\;\tau$}\ \ \ the overall expression $\C[e]$ synthesises:
\[\inference[\texttt{s$\circ_1$}]{\ctxclass{\C}{\synmode\;\tau}{\Delta';\,\Gamma'}{m} & \tau \sqcup (\gap \funarr \gap) \equiv \tau_1 \funarr \tau_2 & \analysis{e_2}{\tau_1}}{\ctxclass{\C(e_2)}{\synmode\;\tau_2}{\Delta';\,\Gamma'}{m}}\]
\[\inference[\texttt{s$\circ_2$}]{\synthesis{e_1}{\tau} & \tau \sqcup (\gap \funarr \gap) \equiv \tau_1 \funarr \tau_2 & \ctxclass{\C}{\anamode\;\tau_1}{\Delta';\,\Gamma'}{m}}{\ctxclass{e_1(\C)}{\synmode\;\tau_2}{\Delta';\,\Gamma'}{m}}\]

\medskip
\fbox{Subsumption: $\synmode\;\tau' \;\leadsto\; \anamode\;\tau$}\ \ \ the unique bridge from synthesis to analysis derivation mode:
\[\inference[\texttt{aSub}]{\ctxclass{\C}{\synmode\;\tau'}{\Delta';\,\Gamma'}{m} & \tau \sim \tau'}{\ctxclass{\C}{\anamode\;\tau}{\Delta';\,\Gamma'}{m}}\]
\caption{Selected context typing rules.}
\label{fig:context-classification}
\end{figure}

Correctness of context typing is formalised via a \textit{totality} theorem: every well-typed expression can be decomposed into context and expression at some focus, with a context typing matching a (consistent) typing derivation of the focused term. Conversely, a \textit{soundness} theorem states that a valid context typing and expression satisfying the focus typing mode can be composed into a well-typed plugged expression.

\begin{majortheorembox}
\begin{theorem}[Plug Decomposition -- Totality]\label{thm:plug-decomposition}
For any context $\C$, expression $e$, and outer derivation mode $d$:
\begin{itemize}
\item If\/ $\synthesis{\C[e]}{\tau}$, then there exist $\Delta'$, $\Gamma'$, and $m$ such that\\ $\ctxclass{\C}{\synmode\;\tau}{\Delta';\,\Gamma'}{m}$ and $e$ satisfies focus mode $m$ under $\Delta';\,\Gamma'$.
\item If\/ $\analysis{\C[e]}{\tau}$, then there exist $\Delta'$, $\Gamma'$, and $m$ such that\\ $\ctxclass{\C}{\anamode\;\tau}{\Delta';\,\Gamma'}{m}$ and $e$ satisfies focus mode $m$ under $\Delta';\,\Gamma'$.
\end{itemize}
Here ``$e$ satisfies focus mode $m$ under $\Delta';\,\Gamma'$'' means: if $m = \synmode\;\tau'$ then $\Delta';\,\Gamma' \vdash e\;\synmode\;\tau'$; if $m = \anamode\;\tau'$ then $\Delta';\,\Gamma' \vdash e\;\anamode\;\tau'$.
\end{theorem}
\end{majortheorembox}

\begin{majortheorembox}
\begin{theorem}[Plug Composition -- Soundness]\label{thm:plug-composition}
For any context $\C$, expression $e$, outer mode $d$, and focus mode $m$: if\/ $\ctxclass{\C}{d}{\Delta';\,\Gamma'}{m}$ and $e$ satisfies focus mode $m$ under $\Delta';\,\Gamma'$, then $\C[e]$ satisfies outer derivation mode $d$ under $\Delta;\,\Gamma$.
\end{theorem}
\end{majortheorembox}

In order to apply the same ideas as synthesis slices, we also prove a substantial static gradual guarantee (\S\ref{sec:core-typing}) for context typings, defining precision on modes $m' \sqsubseteq m$ by lifting: $\synmode\,\tau' \sqsubseteq \synmode\,\tau$ iff $\tau' \sqsubseteq \tau$, and likewise for $\anamode$. The focus mode's class itself (analysis or synthesis) is propagated unchanged through every context typing rule.

\begin{majortheorembox}
\begin{theorem}[Static Gradual Guarantee -- Context, Synthesis Position]\label{thm:graduality-syn-cls}
If\/ $\Gamma' \sqsubseteq \Gamma$, $\C' \sqsubseteq \C$, and\/ $\ctxclass[\Delta;\,\Gamma]{\C}{\synmode\,\tau_p}{\Delta_f;\,\Gamma_f}{m}$, then there exist $\tau_p'$, $\Gamma_f'$, and $m'$ with $\tau_p' \sqsubseteq \tau_p$, $\Gamma_f' \sqsubseteq \Gamma_f$, $m' \sqsubseteq m$, and\/ $\ctxclass[\Delta;\,\Gamma']{\C'}{\synmode\,\tau_p'}{\Delta_f;\,\Gamma_f'}{m'}$.
\end{theorem}
\begin{theorem}[Static Gradual Guarantee -- Context, Analysis Position]\label{thm:graduality-ana-cls}
If\/ $\Gamma' \sqsubseteq \Gamma$, $\C' \sqsubseteq \C$, $\tau_p' \sqsubseteq \tau_p$, and\/ $\ctxclass[\Delta;\,\Gamma]{\C}{\anamode\,\tau_p}{\Delta_f;\,\Gamma_f}{m}$, then there exist $\Gamma_f'$ and $m'$ with $\Gamma_f' \sqsubseteq \Gamma_f$, $m' \sqsubseteq m$, and\/ $\ctxclass[\Delta;\,\Gamma']{\C'}{\anamode\,\tau_p'}{\Delta_f;\,\Gamma_f'}{m'}$.
\end{theorem}
\end{majortheorembox}

\section{Analysis Slices}
\label{sec:analysis}

\providecommand{\anaslice}[2]{\mathit{AnaSlice}\;#1 \mathbin{\blacktriangleleft} #2}
\providecommand{\minanaslice}[2]{\mathit{MinAnaSlice}\;#1 \mathbin{\blacktriangleleft} #2}
\providecommand{\anaposslice}[2]{\mathit{InnerAnaSlice}\;#1 \mathbin{\blacktriangleleft} #2}
\providecommand{\minanaposslice}[2]{\mathit{MinInnerAnaSlice}\;#1 \mathbin{\blacktriangleleft} #2}

\noindent With this in mind, analysis slices are simple analogues to synthesis slices, except being conducted on \textit{analysis inner mode} context typing derivations. And, as the overall programs that we consider must synthesise a type, the usual analysis slices concern derivations only in \textit{synthesis} outer derivation mode.\footnote{An extension to analysis outer modes, which is useful for efficient calculation, is discussed later.}

Intuitively, slices of the context explain why a sub-term is \textit{expected} by the surrounding context to synthesise or analyse against a given type. We typically only care about the cases where the inner mode is in \textit{analysis}; a (well-typed) surrounding context does not actually influence \textit{why} a term synthesises a type.\footnote{A context may define a variable which is used in a synthesising term, non-locally affecting the subterm's type. However, assumption slices are treated entirely \textit{separately} in both modes. That is, we do not ask the analysis slice to explain how exactly the bindings provide the minimal assumptions for the focus, only how it explains the \textit{type} at the focus.}

\subsection{Definition}
\label{sec:ana-slice-def}
Structurally, analysis slices are identical to synthesis slices (\S\ref{sec:syn-slice-def}), but with expression slices $\varsigma$ replaced by context slices $\kappa \in \lfloor\C\rfloor$.

\begin{definition}
\label{def:ana-slice}
An \textit{analysis slice} of $\mathit{Cls} : \ctxclass{\C}{\synmode\,\tau_p}{\Delta';\,\Gamma'}{\anamode\,\tau}$ on $\upsilon \in \lfloor\tau\rfloor$, written $\anaslice{\mathit{Cls}}{\upsilon}$, is a pair $(\kappa, \gamma) \in \lfloor\C, \Gamma\rfloor$ such that $\ctxclass[\Delta;\,\gamma]{\kappa}{\synmode\,\psi}{\Delta'';\,\Gamma''}{\anamode\,\phi}$ for some inner-type witness $\phi \sqsupseteq \upsilon$ (with $\phi \in \lfloor\tau\rfloor$ and $\psi \in \lfloor\tau_p\rfloor$ by graduality, \cref{thm:graduality-syn-cls}).
\end{definition}

\subsection{Minimality}
\label{sec:ana-minimality}

$\minanaslice{\mathit{Cls}}{\upsilon}$ is defined analogously to Definition~\ref{def:isminimal}, minimising over the context slice $(\kappa, \gamma)$. And as before, existence of minimal analysis slices and monotonicity follow from finiteness and graduality.

\begin{majortheorembox}
\begin{theorem}[Existence of minimal analysis slices]\label{thm:ana-min-exists}
For every $s : \anaslice{\mathit{Cls}}{\upsilon}$, there exists $m : \minanaslice{\mathit{Cls}}{\upsilon}$ with $m \sqsubseteq s$.
\end{theorem}
\end{majortheorembox}

\begin{majortheorembox}
\begin{theorem}[Monotonicity of minimal analysis slices]\label{thm:ana-mono}
If $\upsilon_1 \sqsubseteq \upsilon_2$ in $\lfloor\tau\rfloor$ and\\$m_2 : \minanaslice{\mathit{Cls}}{\upsilon_2}$, then there exists $m_1 : \minanaslice{\mathit{Cls}}{\upsilon_1}$ with $m_1 \sqsubseteq m_2$.
\end{theorem}
\end{majortheorembox}

\subsection{Minimally Scoped Analysis Slices}
\noindent As mentioned earlier, a synthesising term derives all of its type information internally.\footnote{Again, given that assumptions are tracked externally.} Therefore, the type is actually only enforced by the context up until the innermost subderivation with a synthesis outer derivation mode.

For example, for an annotated term nested inside another annotation, only the inner annotation actually enforces a type on the inner term. Therefore, an algorithm need only check this region (the outer context can be fully omitted).

\section{Error Marking \& Type Error Debugging}
\label{sec:marking}

This section formally applies type slicing to static error squiggles as in \cref{fig:explain-error} (\S\ref{sec:slicing-by-example}). It demonstrates that we can not only use this calculus to ``explain types'' in well-typed programs, but also effectively ``explain type errors'' in ill-typed programs.

This builds upon the theory of type error marking for bidirectional systems with holes of Zhao et al.~\cite{Marking2024}, which we briefly recap next and extend to the type slicing core calculus (\S\ref{sec:core-calculus}).
 
\subsection{Error Marking}
\label{sec:marking-calculus}

A marking algorithm is a pair of bidirectional judgements,
\[
  \synmark{e}{\chk e}{\tau}
  \qquad\text{and}\qquad
  \anamark{e}{\chk e}{\tau},
\]
reading ``under $\Delta;\Gamma$, the source expression $e$ \textit{marks
to} the marked expression $\chk e$, synthesising (resp.\ being analysed
against) type $\tau$''. The marked output $\chk e$ has the same
structure as $e$ but wraps each point of local type failure in a \emph{mark}. The marks, representing highlighting an entire erroneous term $e$, are tabulated in \cref{tab:mark-forms} alongside the corresponding typing rule failures. 

The core calculus studied gives rise to four classes of errors: \textit{type inconsistencies} when subsumption cannot be applied, \textit{shape errors} when rules expect a sub-term to join (match) a specific type shape, \textit{mode limitations} where syntax which cannot synthesise its type is expected to provide type information; and variable \textit{scoping errors}, which do not pertain to types and so cannot be debugged by type slicing.

\begin{table}[h]
\centering
\small
\begin{tabularx}{\linewidth}{lll : X}
\toprule
\textbf{Mark} & \textbf{Class} & \textbf{Failed Rule(s)} & \textbf{Error} \\
\midrule
$\markincon[e]{\tau'}{\tau}$ & type inconsistency & Subsumption & $e$ synthesises $\tau'$ inconsistent with analysis type $\tau$ \\
$\markarrow[e]$       & shape & Function Application & $e$ is in a function position, but is not a function \\
$\marksum[e]$         & shape & Case Expression & $e$ is a case scrutinee, but is not a sum type \\
$\markprod[e]$        & shape & Product Projection & $e$ is projected as a product, but is not a product \\
$\markforall[e]$      & shape & Type Application & $e$ is in a type function position, but is not a type function \\
$\marklam[\lambda x.\,e]$ & mode limitation & Lambda Synthesis & Unannotated lambda in synthesis position \\
$\marklam[\lambda x.\,e]$ & mode limitation & Lambda Analysis  & Unannotated lambda analysed against a non-function type \\
$\markvar{k}$      & scope & Variables & Variable $k$ has no binding in scope \\
\bottomrule
\end{tabularx}
\caption{Error Mark Forms Classified}
\label{tab:mark-forms}
\end{table}

Importantly, marked terms can be type-checked by treating the marks as `non-empty' holes, synthesising $\gap$ (and recursively marking and type-checking the inner term).

The validity of this method is formalised by two theorems: \textit{totality} ensures marks account for every possible typing error, and \textit{erasure} ensures adding marks does not change a term's structure.

\begin{theorem}[Totality of Marking]
\label{thm:mark-total}

For every expression $e$ and assumptions $\Delta;\Gamma$, there exist a marked expression $\chk e$ and type $\tau$ such that
$\synmark{e}{\chk e}{\tau}$. And, for every type $\tau'$, there exists
$\chk e'$ such that $\anamark{e}{\chk e'}{\tau'}$.
\end{theorem}

\begin{theorem}[Erasure of Marking]
\label{thm:mark-erase}

If $\synmark{e}{\chk e}{\tau}$ then $\mathrm{erase}(\chk e) \equiv e$, and
likewise for analysis.
\end{theorem}

\subsection{Marked Context Typing}
\label{sec:marking-classification}

In this style, we extend context typing (\S\ref{sec:context-typing}) to \textit{marked} context typing $\mctxclass[\Delta;\,\Gamma]{\C}{\chk{\C}}{p}{\Delta';\,\Gamma'}{m}$, threading a marked context $\chk{\C}$ alongside $\C$ (with context marks defined analogously to expressions), adding one error rule per mark form of \cref{tab:mark-forms}. We have proven analogous totality and soundness theorems hold for marked contexts, mirroring the unmarked plug decomposition and composition (\cref{thm:plug-decomposition,thm:plug-composition}).

\subsection{The Debugging Recipe by Example}
\label{sec:marking-recipes}

We may now relate the tools developed over the last four sections to explain type errors:
\begin{enumerate}
\item Place the focus (cursor) at some marked term. This splits the marked program into a marked context $\chk{\mathcal{C}}$ and marked focus $\chk e$, with $\chk{\mathcal{C}}[\chk e]$ recovering the program; a marked context typing at the focus is then derivable by totality of marked context typing.
\item Synthesis slice the internal type information, the term at the \emph{focus}, $\chk e$, if the focus mode $m$ is in synthesis ($\synmode\;\tau$).
\item Analysis slice the external type information, the \emph{context typing} itself, if the focus mode $m$ is in analysis ($\anamode\;\tau$).
\item Pick queries to explore the error. For type-inconsistency marks \textit{both} the synthesis and analysis slices exist (two context typings exist from step 1). From these, \textit{only} the inconsistent parts need be queried to explain the error. For shape errors, simply query just the outermost shape of the synthesised type; mode-limitation errors are analogous on the analysis side.
\end{enumerate}

We give some worked examples, notating with a \markbox{\text{marked term}}, a \sli{\text{synthesis slice}}, and/or an \sli[OliveGreen]{\text{analysis slice}}. `Purely structural' slice regions which do not add or manipulate type information (i.e. bindings) are in a lighter shade.

\subsubsection{Type inconsistency: {$\markincon{\tau'}{\tau}$}}
\label{sec:marking-recipes-incon}

Consider the program
\[
  \mathbf{let}\;p : 1 \funarr 1 = (\texttt{true},\,\texttt{false})\;\mathbf{in}\;p.
\]
The RHS $(\texttt{true},\,\texttt{false})$ synthesises $\texttt{Bool} \times \texttt{Bool}$ but is analysed against $1 \funarr 1$. Marking inserts
$\markincon{\texttt{Bool} \times \texttt{Bool}}{1 \funarr 1}$
around the RHS:
\[
  \mathbf{let}\;p : 1 \funarr 1 = \markbox{\markincon[(\texttt{true},\,\texttt{false})]{\texttt{Bool} \times \texttt{Bool}}{1 \funarr 1}}\;\mathbf{in}\;p.
\]
Marking would ordinarily report the error $\markincon{\texttt{Bool} \times \texttt{Bool}}{1 \funarr 1}$ with the entire term highlighted. However, in reality, only the outermost type constructors are actually causing the (initial) type inconsistency, and type slicing allows query refinements of the full types. In full:

\noindent\paragraph{\textbf{Decomposition}}
The (context, focus) pair the marking factors into is
\[
  \mathcal{C} \;=\; \mathbf{let}\;p : 1 \funarr 1 = \cmark\;\mathbf{in}\;p,
  \qquad
  e \;=\; (\texttt{true},\,\texttt{false}),
\]
with the focus at $\anamode\,(1 \funarr 1)$. The two
debugging queries run on the two sides of this split.

\noindent\paragraph{\textbf{Synthesis Slice}}
$\mathit{MinSynSlice}$ at the minimal type-slice $\upsilon \sqsubseteq \texttt{Bool} \times \texttt{Bool}$ that still witnesses $\upsilon \not\sim (1 \funarr 1)$:
\[
  e \;\rightsquigarrow\; (\,\texttt{true}\,\sli{,}\,\texttt{false}\,)
  \qquad\text{with}\qquad
  \upsilon = \gap\,{\times}\,\gap.
\]

\noindent\paragraph{\textbf{Analysis Slice}}
$\mathit{MinAnaSlice}$ at the minimal $\psi \sqsubseteq 1 \funarr 1$ that still witnesses $(\texttt{Bool} \times \texttt{Bool}) \not\sim \psi$:
\[
  \mathcal{C} \;\rightsquigarrow\;
  \slist{\mathbf{let}}\;p\,\slist{:}\,1\,\sli[OliveGreen]{\funarr}\,1\,\slist{=}\,\cmark\,\slist{\mathbf{in}}\;p
  \qquad\text{with}\qquad
  \psi = \gap\,{\funarr}\,\gap.
\]

\paragraph{\textbf{Full view:}} The slices pick out exactly what each side contributes to the failure.

\queryblock{\texttt{Bool}\sli{\times}\texttt{Bool}\;\;(\equiv\;\gap \times \gap)}{\texttt{1}\sli[OliveGreen]{\funarr}\texttt{1}\;\;(\equiv\;\gap \funarr \gap)}
\[
  \slist{\mathbf{let}}\;p\,\slist{:}\,1\,\sli[OliveGreen]{\funarr}\,1\,\slist{=}\,
  \markbox{\markincon[(\,\texttt{true}\,\sli{,}\,\texttt{false}\,)]{\texttt{Bool} \times \texttt{Bool}}{1 \funarr 1}}\,
  \slist{\mathbf{in}}\;p.
\]

We can see that this recipe does not \textit{arbitrarily} blame either the synthesising term or the analysing context for the error; after all, errors are commonly in annotations, especially when assigning types in a gradual language (49\%)~\cite{TypeAnnotationStudy}.

However, type slicing is flexible. A user who \emph{does} want to trust the surrounding context (annotations), as is often done in conventional type error highlighting, can simply only query the synthesis-side. Or vice versa, to trust the implementation, which is useful when refactoring code to a different type, which conflicts with existing annotations.

\subsubsection{Shape mismatches: $\markarrow$, $\marksum$, $\markprod$, $\markforall$}
\label{sec:marking-recipes-shape}

The shape-mismatch marks all occur in similar situations: the syntactic position of the focus requires a type kind, but the focus's synthesised type does not match this.

\vspace{3pt}\noindent
\begin{minipage}[c]{0.68\linewidth}
\hspace*{9pt}A representative example is $e = \texttt{((), ())}\,(\,\texttt{true}\,)$, where $\texttt{((),())}$ is syntactically in a function position, but synthesises $\texttt{1}\times\texttt{1}$, which is \textit{not} a function, i.e. does not join with $\gap \funarr \gap$. We can therefore query just why it is a product, and not a function.
\end{minipage}\hfill
\begin{minipage}[c]{0.28\linewidth}
\centering
\queryblock{\texttt{1}\sli{\times}\texttt{1}}{\text{---}}
$\markbox{\markarrow[(\texttt{()}\sli,\texttt{()})]}\,(\,\texttt{true}\,).$
\end{minipage}

\subsubsection{Unannotated Lambda against Non-Arrow: $\marklam$}
\label{sec:marking-recipes-lam}

Unannotated lambdas $\lambda y.\,e$ have no synthesis rule; their type can only be
determined by the surrounding context. Consequently it has
no synthesis type, so type inconsistencies are not marked via subsumption, but instead by a mode limitation mark. In this situation all the relevant type information is derived from the surrounding context, so we use analysis slices.

\vspace{3pt}\noindent
\begin{minipage}[c]{0.58\linewidth}
\hspace*{9pt}Consider the program $\mathbf{let}\;b : \texttt{Bool} = \lambda y.\,y\;\mathbf{in}\;b$.
The RHS $\lambda y.\,y$ is a bare lambda analysed against
$\texttt{Bool}$. Taking the $\mathit{MinAnaSlice}$ on the enclosing context witnesses the context enforcing a non-function type (\texttt{Bool}):
\end{minipage}\hfill
\begin{minipage}[c]{0.38\linewidth}
\centering
\queryblock{\text{---}}{\sli[OliveGreen]{\texttt{Bool}}}
$\slist{\mathbf{let}}\;b\,\slist{:}\,\sli[OliveGreen]{\texttt{Bool}}\,\slist{=}\,
  \markbox{\marklam[\lambda y.\,y]}\,\slist{\mathbf{in}}\;b.$
\end{minipage}

\subsubsection{Unbound variable: $\markvar{k}$}
\label{sec:marking-recipes-unbound}

The free-variable mark records a \emph{scope} failure. There are no types involved, so type slicing is inapplicable.

\section{Optimising Type Slice Calculation}
\label{sec:algorithms}

\providecommand{\tmslice}[2]{\mathit{TermMin}\;#1 \mathbin{\blacktriangleleft} #2}
\providecommand{\mintmslice}[2]{\mathit{MinTermMin}\;#1 \mathbin{\blacktriangleleft} #2}
\providecommand{\syntypeslice}[2]{\mathit{TypeSlice}^{\synmode}\;#1 \mathbin{\blacktriangleleft} #2}
\providecommand{\anatypeslice}[2]{\mathit{TypeSlice}^{\anamode}\;#1 \mathbin{\blacktriangleleft} #2}
\providecommand{\tmcalc}[5]{#1 \mathbin{{\color{BrickRed}\blacktriangleleft}} #2 \;{\color{BrickRed}\rightsquigarrow}\; #3 \mathbin{\synmode} #4 \,\dashv\, #5}
\providecommand{\tmrule}[7][\Delta;\,\Gamma]{
  #1 \vdash #2 \mathbin{\synmode} #3
  \mathbin{{\color{BrickRed}\blacktriangleleft}} #4
  \;{\color{BrickRed}\rightsquigarrow}\; #5
  \mathbin{\synmode} #6
  \,\dashv\, #7
}

\providecommand{\stkU}[2]{
  \tikz[baseline=(N.base)]{
    \node[inner sep=0pt, outer sep=0pt] (N) {\strut$#2$};
    \foreach \col/\off in {#1}{
      \draw[\col, line width=1.2pt, line cap=butt, overlay]
        ([xshift=-3pt,yshift=\off]N.south west) -- ([xshift=3pt,yshift=\off]N.south east);
    }
    \path ([yshift=-10pt]N.south west) ([yshift=-10pt]N.south east);
  }
}
\providecommand{\tokensp}{\mskip8mu}
\providecommand{\sABCD}[1]{\stkU{BrickRed/-1pt,NavyBlue/-3.4pt,OliveGreen/-5.8pt,Mulberry/-8.2pt}{#1}}
\providecommand{\sABC}[1]{\stkU{BrickRed/-1pt,NavyBlue/-3.4pt,OliveGreen/-5.8pt}{#1}}
\providecommand{\sABD}[1]{\stkU{BrickRed/-1pt,NavyBlue/-3.4pt,Mulberry/-8.2pt}{#1}}
\providecommand{\sAC}[1]{\stkU{BrickRed/-1pt,OliveGreen/-5.8pt}{#1}}
\providecommand{\sAD}[1]{\stkU{BrickRed/-1pt,Mulberry/-8.2pt}{#1}}
\providecommand{\sBC}[1]{\stkU{NavyBlue/-3.4pt,OliveGreen/-5.8pt}{#1}}
\providecommand{\sBD}[1]{\stkU{NavyBlue/-3.4pt,Mulberry/-8.2pt}{#1}}
\providecommand{\sC}[1]{\stkU{OliveGreen/-5.8pt}{#1}}
\providecommand{\sD}[1]{\stkU{Mulberry/-8.2pt}{#1}}
\providecommand{\sPone}[1]{\stkU{OliveGreen/-1pt}{#1}}
\providecommand{\sPtwo}[1]{\stkU{Mulberry/-3.4pt}{#1}}
\providecommand{\sPboth}[1]{\stkU{OliveGreen/-1pt,Mulberry/-3.4pt}{#1}}

\bigskip
\noindent The theory of synthesis and analysis slices (\S\ref{sec:synthesis}, \S\ref{sec:analysis}) establishes \textit{which} slices exist, and provides a brute-force algorithm. This section explores these structures further to improve the brute-force algorithm.

\subsection{Products of Synthesis Slices}
\label{sec:product-revisited}

We earlier suggested recursively calculating minimal slices of subterms and composing them together (\S\ref{sec:decompositions}). However, this cannot be done naively: the product of two minimal slices is not necessarily minimal. Consider the following if statement:
\[ \Gamma = (x \mathbin{:} \mathbf{1} \funarr \mathbf{1},\; y \mathbin{:} \mathbf{1} \funarr \mathbf{1}), \qquad e = \mathbf{if}\;\mathbf{true}\;\mathbf{then}\;x\;\mathbf{else}\;y, \qquad \tau = \mathbf{1} \funarr \mathbf{1}, \]
and ask for slices of $D : \synthesis{e}{\tau}$ at the full query $\upsilon = \mathbf{1} \funarr \mathbf{1}$. There are \textit{four} incomparable minimal full slices, displayed together as underlines:

\begin{center}
\begin{tikzpicture}[baseline=(t.base), every node/.style={inner sep=3pt}]
  \node[inner sep=0pt] (ax) {$x : \sAC{\mathbf{1}} \mathrel{\sABC{\funarr}} \sBC{\mathbf{1}}$};
  \node[inner sep=0pt, anchor=west] (xmark) at (ax.west) {$\phantom{x : }$};
  \node[inner sep=0pt, right=3.5em of ax] (ay) {$y : \sBD{\mathbf{1}} \mathrel{\sABD{\funarr}} \sAD{\mathbf{1}}$};
  \node[inner sep=0pt, anchor=west] (ymark) at (ay.west) {$\phantom{y : }$};
  \node[inner sep=0pt] at ($(ax.east)!0.3!(ay.west)$) {$,$};
  \node[draw, line width=0.9pt, rounded corners, fit=(ax)(ay),
        label={[font=\scriptsize\itshape]above:assumptions}] (a) {};
  \coordinate (col1) at (xmark.east |- a.south);
  \coordinate (col2) at ($(ax.east |- a.south) + (3pt,0)$);
  \coordinate (col3) at (ymark.east |- a.south);
  \node[font=\tiny\bfseries, BrickRed,   anchor=east, inner sep=1pt] at ([yshift=12pt]col1)  {(A)};
  \node[font=\tiny\bfseries, OliveGreen, anchor=east, inner sep=1pt] at ([yshift=7.2pt]col1) {(C)};
  \node[font=\tiny\bfseries, NavyBlue,   anchor=west, inner sep=1pt] at ([yshift=9.6pt]col2) {(B)};
  \node[font=\tiny\bfseries, NavyBlue,   anchor=east, inner sep=1pt] at ([yshift=9.6pt]col3) {(B)};
  \node[font=\tiny\bfseries, BrickRed,   anchor=west, inner sep=1pt] at ([yshift=12pt]a.south east)  {(A)};
  \node[font=\tiny\bfseries, Mulberry,   anchor=west, inner sep=1pt] at ([yshift=4.8pt]a.south east) {(D)};
  \node[right=5em of a] (t)
       {$\sABCD{\mathbf{if}}\tokensp\mathbf{true}\tokensp\sABCD{\mathbf{then}}\tokensp\sABC{x}\tokensp\sABCD{\mathbf{else}}\tokensp\sABD{y}$};
  \node[font=\tiny\bfseries, BrickRed,   anchor=east, inner sep=1pt] at ([yshift=12pt]t.south west)   {(A)};
  \node[font=\tiny\bfseries, NavyBlue,   anchor=west, inner sep=1pt] at ([yshift=9.6pt]t.south east)  {(B)};
  \node[font=\tiny\bfseries, OliveGreen, anchor=east, inner sep=1pt] at ([yshift=7.2pt]t.south west)  {(C)};
  \node[font=\tiny\bfseries, Mulberry,   anchor=west, inner sep=1pt] at ([yshift=4.8pt]t.south east)  {(D)};
\end{tikzpicture}
\end{center}

\begin{description}
\item[{\color{BrickRed}(A)}] Both branches are retained; the assumption slice keeps $x$'s domain only and $y$'s codomain only (joining to $1 \funarr 1$).
\item[{\color{NavyBlue}(B)}] Symmetric to (A), taking the opposite choices for the assumption slice
\item[{\color{OliveGreen}(C)}] \textbf{Else}-branch omitted.
\item[{\color{Mulberry}(D)}] Symmetric to (C): \textbf{then}-branch omitted.
\end{description}

Slices {\color{BrickRed}(A)} and {\color{NavyBlue}(B)} are minimal because each variable is sliced \textit{differently} across the two uses, and none are \textit{program} slices of each other. However, notice that {\color{BrickRed}(A)} and {\color{NavyBlue}(B)} are not subslices of the product with $x$, where only 2 minima exist:

\begin{center}
\begin{tikzpicture}[baseline=(t.base), every node/.style={inner sep=3pt}]
  \node[inner sep=0pt] (ax) {$x : \sPboth{\mathbf{1}\funarr\mathbf{1}}$};
  \node[inner sep=0pt, anchor=west] (xmark) at (ax.west) {$\phantom{x : }$};
  \node[inner sep=0pt, right=2.5em of ax] (ay) {$y : \sPtwo{\mathbf{1}\funarr\mathbf{1}}$};
  \node[inner sep=0pt, anchor=west] (ymark) at (ay.west) {$\phantom{y : }$};
  \node[inner sep=0pt] at ($(ax.east)!0.3!(ay.west)$) {$,$};
  \node[draw, line width=0.9pt, rounded corners, fit=(ax)(ay),
        label={[font=\scriptsize\itshape]above:assumptions}] (a) {};
  \coordinate (col1) at (xmark.east |- a.south);
  \coordinate (col2) at ($(ax.east |- a.south) + (3pt,0)$);
  \coordinate (col3) at (ymark.east |- a.south);
  \node[font=\tiny\bfseries, OliveGreen, anchor=east, inner sep=1pt] at ([yshift=12pt]col1)  {(C)};
  \node[font=\tiny\bfseries, Mulberry,   anchor=west, inner sep=1pt] at ([yshift=9.6pt]col2) {(D)};
  \node[font=\tiny\bfseries, Mulberry,   anchor=east, inner sep=1pt] at ([yshift=9.6pt]col3) {(D)};
  \node[right=5em of a] (t)
       {$\sPboth{(}\tokensp\sPboth{\mathbf{if}}\tokensp\mathbf{true}\tokensp\sPboth{\mathbf{then}}\tokensp\sPone{x}\tokensp\sPboth{\mathbf{else}}\tokensp\sPtwo{y}\tokensp\sPboth{\mathbf{,}}\tokensp\sPboth{x}\tokensp\sPboth{)}$};
  \node[font=\tiny\bfseries, OliveGreen, anchor=east, inner sep=1pt] at ([yshift=12pt]t.south west)  {(C)};
  \node[font=\tiny\bfseries, Mulberry,   anchor=west, inner sep=1pt] at ([yshift=9.6pt]t.south east) {(D)};
\end{tikzpicture}
\end{center}

When this expression is embedded in the product $(\mathbf{if}\;\mathbf{true}\;\mathbf{then}\;x\;\mathbf{else}\;y,\; x)$, the two assumptions are joined. However, the right projection's use of $x$ is strictly stronger than that of slices {\color{BrickRed}(A)}, {\color{NavyBlue}(B)}, and {\color{Mulberry}(D)}; we are \textit{aliasing} this variable's assumption. Slices {\color{BrickRed}(A)} and {\color{NavyBlue}(B)} are no longer minimal under this aliased assumption, meaning they are not subslices of the overall product. Slice {\color{Mulberry}(D)} is still valid because it does not actually use $x$.

\subsection{Term-Minimal Slices}
\label{sec:term-minimal-slices}

To rule out slices like (A) and (B) (\S\ref{sec:product-revisited}), we strengthen the notion of slices. We first calculate the minimal expression under \textit{maximal} typing assumptions, and then minimise the assumptions afterwards. By graduality, this is a minimal synthesis slice.

\begin{definition}[Term-minimal slice]
\label{def:term-minimal-slice}
A \textit{term-minimal slice} of $D : \synthesis{e}{\tau}$ on $\upsilon \in \lfloor \tau \rfloor$, written $\tmslice{D}{\upsilon}$, is a term slice $\varsigma \in \lfloor e \rfloor$ such that $\Delta;\,\Gamma \vdash \varsigma\;\synmode\;\psi$ for some $\psi \sqsupseteq \upsilon$ (with $\psi \in \lfloor \tau \rfloor$ by graduality, \cref{thm:graduality-syn}).
\end{definition}

Minimality lifts directly from the slice lattice on $e$ alone, and existence and monotonicity follow analogously to synthesis slices (\cref{thm:min-exists,thm:mono}).
\begin{definition}[Minimal Term-Minimal Slices]
\label{def:mintmslice}
A $\mintmslice{D}{\upsilon}$ is a $\tmslice{D}{\upsilon}$ minimal in $\tmslice{D}{\upsilon}$ under term-slice precision on $\lfloor e \rfloor$.
\end{definition}

Crucially, the two notions of minimality do not coincide: term-minimality is strictly stronger. Not all minimal slices are term-minimal under maximal assumptions.

\begin{counterexamplebox}

\begin{counterexample}[Minimal $\not\Rightarrow$ Term-Minimal]\label{lem:min-not-tmin}
Slices (A) and (B) (\S\ref{sec:product-revisited}) are minimal synthesis slices but are \textbf{not} term-minimal.
\end{counterexample}
\end{counterexamplebox}

\subsection{Calculating Term-Minimal Slices}
\label{sec:term-minimal-calculus}
The rest of this section develops a structural calculus that characterises minimal term slices via a judgement:
\[ \Gamma \vdash e \mathbin{\synmode} \tau \mathbin{\blacktriangleleft} \upsilon \;{\color{BrickRed}\rightsquigarrow}\; \varsigma \mathbin{\synmode} \psi \dashv \gamma, \]
read ``the typing $\Gamma \vdash e \synmode \tau$ at query $\upsilon$ produces term slice $\varsigma$, synthesises $\psi$, and uses assumptions $\gamma$''. The fourth output $\gamma$ is a minimal assumption slice to retrieve a minimal slice corresponding to the term-minimal slice.

\noindent The cases that deserve commentary are products, the binding constructs, function application, and type application.

\subsubsection{Products}
With term minimality, we now have that two minimal synthesis slices which are also term-minimal can be composed into a term-minimal, minimal synthesis slice. For a query $\upsilon_1\times \upsilon_2$ on the product, the projections are queried with $\upsilon_1$ and $\upsilon_2$ individually:

\[
\inference[\synr{Pair}]{
  \tmrule{e_1}{\tau_1}{\upsilon_1}{\varsigma_1}{\psi_1}{\gamma_1}
  &
  \tmrule{e_2}{\tau_2}{\upsilon_2}{\varsigma_2}{\psi_2}{\gamma_2}
}{\tmrule{(e_1, e_2)}{\tau_1 \times \tau_2}{\upsilon_1 \times \upsilon_2}{(\varsigma_1, \varsigma_2)}{\psi_1 \times \psi_2}{\gamma_1 \sqcup \gamma_2}}
\]

\subsubsection{Function Application}
\label{sec:term-min-app}

All the synthesised type information from a function application is derived from the function itself. We may simply omit the argument, and slice the function by $\gap \funarr \upsilon$:
\[
\inference[\synr{App}]{
  \analysis{e_2}{\tau_1}
  &
  \upsilon \neq \gap
  &
  \tmrule{e_1}{\tau_1 \funarr \tau_2}{\gap \funarr \upsilon}{\varsigma_{\mathit{fn}}}{\psi_1}{\gamma}
}{\tmrule{e_1\,e_2}{\tau_2}{\upsilon}{\varsigma_{\mathit{fn}}(\gap)}{\mathsf{cod}(\psi_1)}{\gamma}}
\]

\subsubsection{Binding Constructs}
\label{sec:term-min-binders}

Rules which bind variables are more complex, as they abstract a variable used in the minimised body. The general method is to slice them by first slicing the \textit{body}, extracting the assumption it needed on the bound variable from the minimal assumptions $\gamma$, then slice the definition / type annotation using this. 

Notice that this is the \textit{opposite} direction to typechecking, e.g. let bindings type the definition first, then use the result to type the body, but slicing slices the body first to determine the minimal assumption to slice the definition with.

\subsubsection{Type Application}
\label{sec:term-min-tapp}

Type application works by matching the type function's type against the query to collect a list of constraints on the type variable, the \textit{join} of which gives the sliced type argument $\hat\sigma$. The type function is then sliced with a matched query $\widehat{\upsilon}$, where each position corresponding to the type variable $\alpha$ in the type function's type is replaced by $\alpha$ itself and other positions are carried over from the query unchanged. Figure~\ref{fig:tapp-minsub} illustrates the construction on a query whose two subparts are matched with an $\alpha$ with differing constraints.

\begin{figure}[H]
\centering
$e \mathbin{\synmode} \forall\alpha.\,(\alpha \times \alpha \times \mathbf{Bool})$,\quad applied at\quad $\sigma = \mathbf{1} \funarr (\mathbf{Bool} \times \mathbf{Bool})$:

\medskip

So $e\langle\sigma\rangle \mathbin{\synmode} (\mathbf{1} \funarr (\mathbf{Bool} \times \mathbf{Bool})) \times (\mathbf{1} \funarr (\mathbf{Bool} \times \mathbf{Bool})) \times \mathbf{Bool}$.

\medskip

\begin{tabular}{r@{\;\;}l}
\textbf{Query}             & $\upsilon \;=\; \underbrace{\sli[teal]{\mathbf{1} \funarr \gap}}_{\alpha} \;\times\; \underbrace{\sli[teal]{\gap \funarr (\mathbf{Bool} \times \gap)}}_{\alpha} \;\times\; \gap$ \\[16pt]
\textbf{Matched query}     & $\widehat{\upsilon} \;=\; \alpha \;\times\; \alpha \;\times\; \gap$ \\[8pt]
\textbf{Constraints}       & $\mathbf{1} \funarr \gap \;\sqsubseteq\; \alpha \qquad \gap \funarr (\mathbf{Bool} \times \gap) \;\sqsubseteq\; \alpha$ \\[8pt]
\textbf{Constraints Join}  & $\hat\sigma \;=\; \mathbf{1} \funarr (\mathbf{Bool} \times \gap)$ \\
\end{tabular}
\caption{Type application query constraints.}
\label{fig:tapp-minsub}
\end{figure}

\subsection{Case Expressions}
\label{sec:fixed-point}

The case rule is the difficult one, appearing to require an iterative fixed point algorithm. This is because the two branch types join to produce the resulting type which is being queried upon. As synthesis slices are not exact, a query may be split into two disjoint types, yet a minimal slice of one branch may produce a more precise type than the query (overflow), which is also produced by a sibling branch. The sibling branch may then admit further slicing, as it no longer needs to provide this redundant part of the type. Lack of an upwards version of monotonicity (\cref{fig:syn-slice-refinement}) makes this difficult to solve in general without brute force.

However, when the branch slices exactly synthesise disjoint type queries, then there is no redundancy and the result is minimal. The scrutinee can be sliced using the join of the minimal assumptions on the branches, much like with let bindings. 

To express disjoint types we can use \textit{co-Heyting subtraction} $\upsilon_1 \heytingminus \upsilon_2$, which is the least slice whose join with $\upsilon_2$ recovers $\upsilon_1$, i.e. ``the part of $\upsilon_1$ not already covered by $\upsilon_2$'', e.g. $\sli{1 \times 1} \;\heytingminus\; \sli{1 \times}\,\gap \;=\; \gap\,\sli{\times 1}$.

Then, a disjoint cover of query $\upsilon$ is one such that $\upsilon\heytingminus \upsilon_1 = \upsilon_2$ and $\upsilon \heytingminus \upsilon_2 = \upsilon_1$. The case expression will necessarily synthesise $\upsilon_1 \sqcup \upsilon_2$ which equals $\upsilon$ by standard co-Heyting algebra rules. Note that the lattice is \textit{not} a boolean algebra, so $\upsilon_1 \sqcap \upsilon_2 \neq \gap$ in general.

The Hazel implementation simply does a disjoint split and ignores the overflow issue to work in linear time complexity. This is why it is an approximation. But, it is perfectly feasible to use this approximation and then continue with the brute force algorithm, as we expect the approximation to be very close to the minimum in almost all cases. Further, as overflow only occurs due to variable aliasing, the brute force search need only consider aliased variables at the leaves.

\subsection{Calculating Analysis Slices}
\label{sec:ana-slice-calculus}

Analysis slices are calculated directly from synthesis slices. The only rule that creates a type checking judgement is function application. Consider the analysis slice of a function argument with query type $\upsilon_1$; it is just the function application and a slice of the function under $\gap \to \upsilon_1$. That is, to explain why a function argument must be a given type, we ask why the corresponding function has that type as its domain.

However, to recursively calculate the analysis slice of an application when the focus is deep inside the argument, the analysis slice of the argument itself must be calculated, but its outer derivation mode is in \textit{analysis}. The analysis slices defined earlier (\S\ref{sec:ana-slice-def}) do not apply here.

We need an additional notion of \textit{inner} analysis slices. Dually to regular (outer) analysis slices, which treat the synthesis mode as an \textit{input}, the inner slices treat the analysis derivation mode as an \textit{output}, and attempt to minimise it. 

Intuitively, an inner analysis slice gives the minimal assumptions, context, and additionally the minimum \textit{analysis type} to be imposed on it by the outer context, such that the queried analysis type holds at its focus. In particular, slicing the function in an application would use this minimal outer analysis mode type on the argument to slice the domain type of the function.
\vspace{-0.5em}
\begin{definition}
\label{def:ana-pos-slice}
An \textit{inner analysis slice} of $\mathit{Cls} : \ctxclass{\C}{\anamode\,\tau_p}{\Delta';\,\Gamma'}{\anamode\,\tau}$ on $\upsilon \in \lfloor\tau\rfloor$, written $\anaposslice{\mathit{Cls}}{\upsilon}$, is a triple $(\kappa, \gamma, \upsilon_{\mathrm{outer}}) \in \lfloor\C, \Gamma, \tau_p\rfloor$ such that
\[
\ctxclass[\Delta;\,\gamma]{\kappa}{\anamode\,\upsilon_{\mathrm{outer}}}{\Delta'';\,\Gamma''}{\anamode\,\phi}
\]
for some inner-type witness $\phi \sqsupseteq \upsilon$ (with $\phi \in \lfloor\tau\rfloor$ by graduality, \cref{thm:graduality-ana-cls}).
\end{definition}
\vspace{-0.5em}

\noindent A minimal inner analysis slice minimises the whole triple, $\upsilon_{\mathrm{outer}}$ included (so that an application slices the function's domain by no more than the argument requires).

\section{Full Type Slices -- Slicing Assumption Definitions}
\label{sec:full-slices}
Synthesis slices are explained using a set of \textit{minimal assumptions}. However, in real code, these are actually \textit{provided by the surrounding context} (or from external code). Providing these minimal assumptions via a full slice of this context can be very meaningful. For example, when debugging a type error, the error itself may be \textit{non-local}, within a variable definition elsewhere, giving one of the used variables the wrong type. In this situation the error is in the definition, and we could debug it by inspecting (and slicing) that definition.

A \textit{full type slice} is a fully sliced decomposition of a program around a focus, which can be pieced back together to form a well-typed program with the same (or stronger) typing properties at the focus as some given query:

\vspace{-0.5em}
\begin{definition}[Type slices]
  \label{def:type-slice}
A \textit{type slice} of $\mathit{Cls} : \ctxclass{\C}{\synmode\,\tau_p}{\Delta';\,\Gamma'}{\synmode\,\tau}$ and $D : \synthesis[\Delta';\,\Gamma']{e}{\tau}$ on $\upsilon \in \lfloor\tau\rfloor$, written $\syntypeslice{(\mathit{Cls},D)}{\upsilon}$, is a pair $((\kappa,\gamma),(\gamma_f,\varsigma)) \in \lfloor\C,\Gamma\rfloor \times \lfloor\Gamma',e\rfloor$ such that $(\gamma_f,\varsigma) : \synslice{D}{\upsilon}$, and $\ctxclass[\Delta;\,\gamma]{\kappa}{\synmode\,\psi}{\Delta';\,\gamma_f'}{\synmode\,\phi'}$ for some $\psi$, and $\gamma_f' \sqsupseteq \gamma_f$ such that $\synthesis[\Delta';\ \gamma_f']{\varsigma}{\phi'}$.
\end{definition}
\vspace{-0.5em}

Importantly, the pair of slices forms a well-typed program: by context typing soundness, $\kappa[\varsigma]$ synthesises $\psi$. Then, a \textit{minimal} type slice is one in which the context $\kappa$ and the focused expression $\varsigma$ are jointly minimised alongside minimal external typing assumptions of the context typing $\gamma$. Notably, we do \textit{not} minimise the assumptions to the focused term, as its assumptions are now explicitly provided by the surrounding context (and $\gamma$).

As slices can be inexact, efficiently calculating a minimal context providing sufficient assumptions to a minimal \textit{SynSlice} at the focus is non-trivial. The Hazel implementation instead closely approximates this by slicing the definition / annotation on each binding by its required type as per the minimal assumptions of the focused term. When each slice of these definitions is exact, the resulting type slice is minimal.

For terms in analysis mode, the full type slice is just the analysis slice, as the empty focus (with type $\gap$) analyses successfully against any type.

\section{Minimum-Size Slices}
\label{sec:min-size}

Minimal slices are not unique, and some are actually larger than others in the sheer number of terms highlighted. So, one might rank them by \textit{size}: the number of highlighted positions, asking for a \textit{minimum-size} slice. For this core calculus, finding such a slice is \textbf{NP-hard} in general. This implies that an algorithm attempting to also reduce the absolute size of slices should use heuristics to avoid or reduce the likelihood of cases with exponential time complexity.

We can construct a reduction from the NP-hard set cover optimisation problem as follows, and shown in \cref{fig:nphard}:
\begin{itemize}
  \item Encode a universe $U = \{1,\dots,n\}$ as a type $\tau = 1 \times \cdots \times 1$ ($n$ factors, one position per element)
\item Encode each of $m$-many sets $S_j$ in the set cover problem by an assumption $x_j : T_j \sqsubseteq \tau$ in $\Gamma$ with $T_j$ defined by the $i$th projection of $T_j$ being $1$ if and only if $i \in S_j$.
\item Construct a case tree with empty scrutinees and placing all $(x_j, I_j)$, defining $I_j$ as the $m$-ary product with $j$th projection being $1$ and all others $\gap$.
\end{itemize}

Then, if a synthesis slice is queried on this case tree with $\tau \times 1 \times \dots \times 1$, each $I_j$ must be retained, so the whole tree structure must be retained, and the only variance in the size of the slices is purely dependent on which variables $x_j$ were retained.

As we queried $\tau$ in the first projection then $\bigsqcup_{j=1}^m{T_j} = \tau$, which implies there are $m$ subsets\footnote{Subsets, as a minimal slice can use less precise assumptions $x_j$, i.e. $x_j : T'_j \sqsubseteq T_j$} $S_j' \subseteq S_j$ such that $\bigcup_{j=1}^mS'_j = U$. Hence, the retained $S_j$ form a set-cover. Finally, the only variance in the size of the slice was in the number of variables $x_j$ chosen, i.e. the number of sets chosen; hence, a minimum-size synthesis slice corresponds to a minimum set cover.

\begin{figure}[ht]
\centering
\begin{tikzpicture}[scale=0.92, transform shape,
    sc/.style={draw, rounded corners=2pt, inner sep=2pt, minimum height=4mm, font=\footnotesize},
    sel/.style={fill=OliveGreen!18, draw=OliveGreen!60!black, line width=0.7pt},
    no/.style={fill=black!5, draw=black!40},
    ca/.style={font=\bfseries\footnotesize, inner sep=1pt},
    lf/.style={draw, rounded corners=1pt, inner sep=2pt, minimum width=5mm, minimum height=4mm, font=\small},
    kp/.style={draw=OliveGreen!60!black, fill=OliveGreen!18, line width=0.7pt},
    dr/.style={draw=black!40, fill=black!4, text=black!50},
    every node/.append style={font=\small},
    line cap=round]

\node[font=\itshape\small] at (-0.6,0.6) {$U$};
\node[sc] at (0.0,0.6) {1};
\node[sc] at (0.8,0.6) {2};
\node[sc] at (1.6,0.6) {3};
\node[sc] at (2.4,0.6) {4};

\node[font=\itshape\small] at (-0.6,0.0) {$S_1$};
\node[sc, sel, minimum width=14mm] at (0.4,0.0) {\{1,2\}};
\node[font=\itshape\small] at (-0.6,-0.6) {$S_2$};
\node[sc, no,  minimum width=14mm] at (1.2,-0.6) {\{2,3\}};
\node[font=\itshape\small] at (-0.6,-1.2) {$S_3$};
\node[sc, sel, minimum width=14mm] at (2.0,-1.2) {\{3,4\}};
\node[font=\itshape\small] at (-0.6,-1.8) {$S_4$};
\node[sc, no,  minimum width=5mm]  at (0.0,-1.8) {1};
\node[sc, no,  minimum width=5mm]  at (2.4,-1.8) {4};

\node[font=\footnotesize\itshape, text=gray!50!black, align=center] at (1.2,-2.5)
  {min cover: $\{S_1,S_3\}$};

\node[font=\footnotesize, align=left] at (5.0,0.6)  {$\tau \;=\; 1 \times 1 \times 1 \times 1$};
\node[font=\footnotesize, align=left] at (5.0,0.0)  {$T_1 \,=\, 1 \times 1 \times \gap \times \gap$};
\node[font=\footnotesize, align=left] at (5.0,-0.6) {$T_2 \,=\, \gap \times 1 \times 1 \times \gap$};
\node[font=\footnotesize, align=left] at (5.0,-1.2) {$T_3 \,=\, \gap \times \gap \times 1 \times 1$};
\node[font=\footnotesize, align=left] at (5.0,-1.8) {$T_4 \,=\, 1 \times \gap \times \gap \times 1$};
\node[font=\footnotesize\itshape, text=gray!50!black, align=center] at (5.0,-2.5)
  {encoding $S_j \to T_j$};

\node[ca] (rt) at (10.5, 0.8) {case};
\node[ca] (cl) at ( 9.1, 0.0) {case};
\node[ca] (cr) at (11.9, 0.0) {case};
\node[lf, kp] (l1) at ( 8.4,-0.8) {$\ell_1$};
\node[lf, dr] (l2) at ( 9.8,-0.8) {$\gap$};
\node[lf, kp] (l3) at (11.2,-0.8) {$\ell_3$};
\node[lf, dr] (l4) at (12.6,-0.8) {$\gap$};
\draw (rt) -- (cl);
\draw (rt) -- (cr);
\draw (cl) -- (l1);
\draw (cl) -- (l2);
\draw (cr) -- (l3);
\draw (cr) -- (l4);
\node[font=\scriptsize, text=BrickRed, anchor=south] at (rt.north) {$T_1 \sqcup T_3 = \tau$};
\node[font=\footnotesize\itshape, text=gray!50!black, align=center] at (10.5,-1.7)
  {kept $J = \{x_1, x_3\}$};
\node[font=\footnotesize, align=center] at (10.5,-2.8)
  {$\ell_1 = (x_1,\; {1} \times \gap \times \gap \times \gap)$\\
   $\ell_3 = (x_3,\; \gap \times \gap \times {1} \times \gap)$};
\end{tikzpicture}
\caption{NP-Hardness of Minimum-Size Slicing}
\label{fig:nphard}
\end{figure}
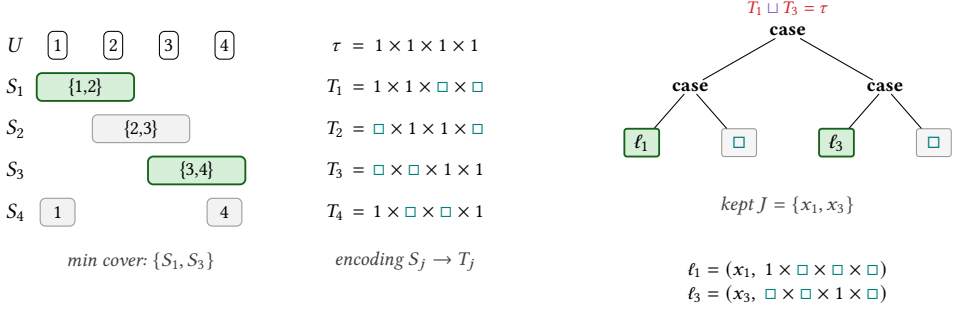

\section{Related Work}
\label{sec:comparisons}
\label{sec:related-work}

Program slicing methods have been explored extensively, in various forms and for many applications~\cite{TipSurvey1995}. Here we compare type slicing with these, being loosely related in foundation to \emph{Galois slicing}~\cite{Perera2012}, and in application to \emph{type error slicing}~\cite{Wand1986,HaackWells2003}, before relating to the broader type debugging literature (\S\ref{sec:cmp-debuggers}).

\subsection{Classical Program Slicing}
\label{sec:cmp-program-slicing}

    Classical program slicing~\cite{TipSurvey1995} was first introduced by Weiser~\cite{Weiser1984}. Here, a slice is a subprogram preserving the value of a chosen variable at a chosen program point: a `criterion'. Similarly to type slicing, the subprogram is a syntactically well-formed program which can be analysed (evaluated).
The criterion is superficially similar to type slicing queries, with two significant differences. First, a type query is itself a slice of a type, which has \textit{compound} structure, rather than a set of plain variables. Second, type slices cannot in general \textit{preserve} all queries exactly (\cref{lem:exact-not-exist}), so we ask only for \textit{validity} (\cref{def:syn-slice}), which remains a `sufficient' explanation.

Classical slicing distinguishes \textit{static} slices (preserved over all executions) from \textit{dynamic} slices (over one fixed execution)~\cite{KorelLaski1988}. Type slicing can be viewed as a \textit{dynamic} slicing method, as it fixes a particular typing derivation; however, most type systems, including ours, have no non-trivial non-determinism, making this distinction largely irrelevant.

\subsection{Galois Slicing}
\label{sec:cmp-galois}

Galois slicing~\cite{Perera2012} is a lattice-theoretic approach to dynamic program slicing that has been applied and related to a wide range of calculi and user interfaces~\cite{Ricciotti2017,PereraGarg2016,PereraVis2022,AtkeyPerera2025}. Type slicing shares the foundational use of lattices and precision, but differs decisively in the lack of a \textit{Galois connection}. 

In Perera's setting, a program evaluation lifts to a \textit{meet-preserving} map $\mathit{eval}_{\rho,e} : \lfloor \rho, e\rfloor \to \lfloor v\rfloor$ between the slice lattices of partial programs and partial values. By an adjoint functor theorem~\cite[Prop.~7.34]{DaveyPriestley}, such a map is the upper adjoint of a unique Galois connection, whose lower adjoint $\mathit{uneval}$ is the `backward slice function' sending a partial-value query to a \textit{least} program slice evaluating to it.

For total type systems assigning unique types like the core calculus (\S\ref{sec:core-calculus}), the natural type-slicing analogue of $\mathit{eval}_{\rho,e}$ maps partial assumptions and expressions to partial types by type checking.
\[
  \phi_D : \lfloor \Gamma, e \rfloor \to \lfloor \tau \rfloor
\]

However, this `forward map' $\phi_D$ is \textit{not meet-preserving} in general. Consider $D : \synthesis{(\mathbf{case}\;e\;\mathbf{of}\;x_1.\,e_1 \mid x_2.\,e_2)}{\tau}$ in which both branches independently synthesise the same type $\tau \sqsupset \gap$, and take the two slices $s_1$ keeping only the first branch and $s_2$ keeping only the second. Both are defined, yielding $\phi_D(s_1) = \phi_D(s_2) = \tau$, but their meet $s_1 \sqcap s_2$ omits both branches, giving $\phi_D(s_1 \sqcap s_2) = \gap \neq \tau = \phi_D(s_1) \sqcap \phi_D(s_2)$.

\subsection{Type Error Slicing}
\label{sec:cmp-tes}

Type error slicing, like type slicing, is a static analysis on a typing problem. Both share the broad goal of attributing type-level information back to source positions, but differ in language target, the underlying mathematical structure, and the scope.

Wand~\cite{Wand1986} and Haack and Wells~\cite{HaackWells2003} work in a globally inferred Hindley-Milner system; the program is reduced to a system of equality constraints over type variables. Wand's instrumented unification reports the constraint sites contributing to a failure, with no minimality guarantee; Haack and Wells compute a type error slice as a \emph{minimal} unsatisfiable subset of the constraint set, projected back to the corresponding minimal subset of source positions. Skalpel~\cite{Skalpel2017} matures and scales this approach to Standard ML, and the Chameleon debugger~\cite{StuckeyChameleon2003} underlines the same minimal location sets interactively.

These slices are not structure-preserving programs but flat sets of constraints and sources, so the results cannot themselves be type checked or executed. Tip and Dinesh~\cite{TipDinesh2001} are the exception: slicing explicitly typed languages by dependence tracking in a rewriting-based checker, their slices are programs guaranteed to reproduce the same error when type checked.

The scope also differs: type error slicing answers \emph{``why does this program fail to type-check?''}, whereas type slicing answers \emph{``why does this expression have this type?''}. Via marking (\S\ref{sec:marking})~\cite{Marking2024}, which elaborates every well-formed expression into a totally typed derivation, type slicing applies also to ill-typed programs. So, here, the question of ``why this type?'' subsumes ``why this fails to type-check?''.

\subsection{Type-Directed Slicing}
\label{sec:cmp-type-directed}

\citet{TypeDirectedSlicing} recently independently introduced \emph{type-directed slicing}, implemented for Java in the tool \textsc{Specimin}. Their motivation and application is more specific than ours: their slices serve the \emph{maintainers of a typechecker}, reducing a large program exhibiting a known bug to a small test case reproducing it.

Given a target program element and a typechecker specified by its type rules, a type-directed slicer statically computes a smaller program on which the typechecker behaves identically at the target: essentially type slicing restricted to maximal queries (no partial queries). However, the approach, theory, and proofs differ greatly, as do the languages considered, with ours extending to bidirectionally inferred types, applying to programs with fewer annotations and better support for first-class functions. On the other hand, their work applies in practice to a wider range of analyses, including exceptions, crashes, and nullness analysis.

Their work was empirically evaluated on historical bugs of \texttt{javac}, NullAway~\cite{NullAway}, and the Checker Framework~\cite{CheckerFramework}, suggesting our theory is similarly applicable to such real world use cases; conversely, partial queries may be a worthwhile extension to their tools, especially when a type checker produces multiple cascading bugs.

\subsection{Type Debuggers \& Type Explanation}
\label{sec:cmp-debuggers}

Beyond slicing, there is a broad literature on explaining, localising, and repairing type errors. The oldest line \emph{explains} types: Beaven and Stansifer~\cite{BeavenStansifer1993} and Duggan and Bent~\cite{DugganBent1996} generate textual explanations from instrumented unification, \citet{YangMichaelsonTrinder2002} explain inferred polymorphic types of well-typed programs, and \citet{Chitil2001} navigates compositional explanation graphs via algorithmic debugging. These systems answer the same question as a synthesis slice---why does this term have this type?---but produce prose or inference traces, whereas a type slice is itself a minimal well-formed program that can be re-checked and further refined. Interactive debuggers similarly answer such queries: Chameleon~\cite{StuckeyChameleon2003} over a constraint store for Haskell, reporting errors as minimal location sets as in \S\ref{sec:cmp-tes}, and Argus~\cite{Argus2025} for Rust trait errors in the IDE.

A larger body of work improves error \emph{reports} for globally inferred languages: constraint-solving heuristics in Helium~\cite{HeerenHageSwierstra2003}, ranking likely error causes over constraint graphs~\cite{ZhangMyers2014,SHErrLoc2015}, minimum error sources via MaxSMT~\cite{PavlinovicKingWies2014,PavlinovicICFP2015}, black-box correcting sets~\cite{Mycroft2016}, subtyping-constraint flow messages~\cite{FlowMessages2023}, dynamic witnesses~\cite{SeidelJhalaWeimer2016}, and learned localisation~\cite{NATE2017} and repair~\cite{Rite2020}. These techniques answer ``where is the error?'' or ``how might it be fixed?'' rather than ``why this type?''; empirically, most type errors implicate multiple program points~\cite{StudentTypeErrorFixes}, which type slicing surfaces directly.

As slices are themselves programs, these tools compose with type slicing (\S\ref{sec:intro}).

\section{Further Extensions \& Applications}
\label{sec:further-work}

\subsection{Set-Theoretic Types}
\label{sec:set-theoretic}
\label{sec:future-set-theoretic-occurrence}

Type slicing looks especially promising for common gradual type systems in which \textit{control flow and types are intertwined}: an expression's type depends non-trivially on the branches inside a function, or the control-flow path in which the expression lies. Explaining a type then genuinely requires considering the branches within the code, rather than just annotations and literals; as type slices necessarily \textit{retain} program structure, they can retain the appropriate control flow structure.

\emph{Set theoretic types} extend the type language with union, intersection, and negation, letting a term carry several types simultaneously, typing overloaded functions and typecase constructs common in real world hybrid typed languages (e.g.\ TypeScript, or Elixir~\cite{ElixirTypes}). Relevant foundations, including graduality proofs, have been explored~\cite{CastagnaLanvin, CastagnaLPS, AngeloThesis}, ready for type slicing to build upon.\footnote{After adapting into a bidirectional setting with expression holes.}

To demonstrate this application, consider generalising sum types of the core calculus (\S\ref{sec:core-calculus}) to allow sum injections to synthesise partially dynamic types: when $\synthesis{e}{\tau}$, let $\synthesis{\iota_1\; e}{\tau + \gap}$ and $\synthesis{\iota_2\; e}{\gap + \tau}$.
Then an if statement $\mathit{pred} \triangleq \mathbf{if}\ x > 0\ \mathbf{then}\ \iota_1(x - 1)\ \mathbf{else}\ \iota_2(())$ synthesises a sum type $\mathbb{N} + 1$ without requiring annotations, much as union type inference works (i.e. union type $\mathbb{N}\ |\ 1$), and type slicing can explain where each component comes from: querying $\gap + 1$ highlights exactly the failing branch, and $\mathbb{N} + \gap$ the succeeding one. We mechanised this example.
\subsection{Occurrence Typing}
\label{sec:occurrence}

\emph{Occurrence typing}~\cite{OccurrenceTyping} is another common technique in hybrid type systems. Inside the true branch of \texttt{if (typeof x = int) ... else ...}, the type of \texttt{x} is refined as an \texttt{int}. These explicitly require exploring the surrounding control flow, which can be extensive, to reason about the types. As type slices necessarily retain program structure, they can provide a reduced slice with minimal control flow to reason within. Formally applying type slicing to such systems will, however, need extensive integration work with bidirectional typing and static graduality.

\subsection{Implicit Polymorphism}
\label{sec:future-implicit}

Globally inferred languages have much more confusing typing derivations and errors than locally inferred ones, so extending the calculus to implicit polymorphism would expand the application significantly. Explaining even well-typed code is particularly useful here: error marks are often mis-localised, leaving the real error in well-typed code whose types conflict with the programmer's expectation. Adapting the constraint slicing of Wand~\cite{Wand1986} and Haack and Wells~\cite{HaackWells2003} to integrate with type slicing is worth exploring.

\subsection{Cast Slicing \& Dynamics}
\label{sec:cast-dynamics}

\noindent
\begin{minipage}[t]{0.58\linewidth}
\vspace{0pt}
A gradually-typed language elaborates programs into a cast calculus and runs the casts dynamically. Type slices can be attached to those casts, letting casts explain their insertion by pointing back to code: the source is a slice of the term being cast, the target a slice of its context.

\hspace*{9pt}This is useful for debugging dynamic type errors, which would not otherwise point back to code: gradual typing's blame~\cite{WadlerFindler2009} identifies \emph{which} cast failed but implicates only the cast site; a slice widens this to every source fragment contributing to either side.
\end{minipage}\hfill
\begin{minipage}[t]{0.38\linewidth}
\vspace{0pt}
\centering
\includegraphics[width=0.85\linewidth]{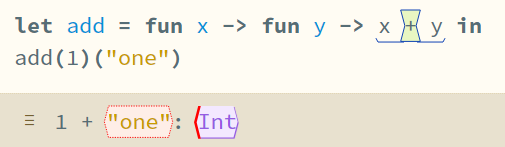}
\captionof{figure}{Type slice associated with a cast error, from a preliminary workshop paper~\hatraSelfCite{} built on differing foundations~\priorDissertationSelfCite{}.}
\label{fig:cast-slice}
\end{minipage}

\vspace{3pt}

However, slices are no longer slices of the original program once dynamic reductions are performed. Tracking dynamic execution, perhaps integrating ideas from Perera~\cite{Perera2012,PereraGarg2016}, would allow tracing casts back to source with stronger guarantees; dynamic slices measurably aid diagnosis of run-time type errors in gradual languages~\cite{TypeSlicer2025}.

\section{Conclusions}
\label{sec:conclusions}

This paper formulated type slicing theory for a bidirectional calculus with explicit System~F polymorphism, bindings, products, and sums, mechanised in Agda, and formally integrated it with the error-marking calculus of \citet{Marking2024}, applying it to ill-typed programs and justifying its \textit{sound} use in debugging \textit{all} static type errors. This positions type slicing as a promising tool for practical languages, especially gradually typed ones, explaining both static and, in follow-up work, \textit{dynamic} types.

\label{lastcontentpage}

\end{document}